\documentclass[11pt, twoside]{amsart}

\usepackage[utf8]{inputenc}
\usepackage[T1]{fontenc}
\usepackage{amssymb,amsmath,amstext}
\usepackage{hyperref, graphicx}
\usepackage{comment}
\usepackage{mathtools}

\newtheorem{theor}{Theorem}[section]
\newtheorem{lem}[theor]{Lemma}
\newtheorem{defin}[theor]{Definition}

\newtheorem{exam}[theor]{Example}

\newtheorem{rem}[theor]{Remark}

\numberwithin{equation}{section}

\newcommand{\mr}{\mathrm}

\newcommand{\es}{\emptyset}

\newcommand{\mcA}{\mathcal{A}}

\newcommand{\mfs}{\mathfrak{s}}

\newcommand{\mbE}{\mathbf{E}}

\newcommand{\mbW}{\mathbf{W}}
\newcommand{\mbX}{\mathbf{X}}
\newcommand{\mbY}{\mathbf{Y}}

\newcommand{\mbbP}{\mathbb{P}}
\newcommand{\mbbN}{\mathbb{N}}

\newcommand{\msfD}{\mathsf{D}}

\newcommand{\msfC}{\mathsf{C}}

\newcommand{\rng}{\mathrm{rng}}

\title[A convergence law for continuous logic]{A convergence law for continuous logic and continuous structures with finite domains}

\author{Vera Koponen}

\address{Vera Koponen, Department of Mathematics, Uppsala University, Sweden.}
\email{vera.koponen@math.uu.se}

\date{26 February, 2026}

\begin{document}

\begin{abstract}
We consider continuous relational structures with finite domain $[n] := \{1, \ldots, n\}$
and a many valued logic, $CLA$,  with values in the unit interval and which uses continuous connectives and 
continuous aggregation functions.
$CLA$ subsumes first-order logic on ``conventional'' finite structures. 
To each relation symbol $R$ and identity constraint $ic$ on a tuple the length of which matches the arity of $R$ we associate a 
continuous probability density function $\mu_R^{ic} : [0, 1] \to [0, \infty)$. 

We also consider a probability distribution on the set $\mbW_n$ of continuous structures with domain $[n]$ which is such that
for every relation symbol $R$, identity constraint $ic$, and tuple $\bar{a}$ satisfying $ic$, the distribution of 
the value of $R(\bar{a})$ is given by $\mu_R^{ic}$, independently of the values for other relation symbols or other tuples.

In this setting we prove that every formula in $CLA$ is asymptotically equivalent to a formula without any aggregation function.
This is used to prove a convergence law for $CLA$ which reads as follows for formulas without free variables:
If $\varphi \in CLA$ has no free variable and $I \subseteq [0, 1]$ is an interval, then there is $\alpha \in [0, 1]$ such that,
as $n$ tends to infinity, the probability that the value of $\varphi$ is in $I$ tends to $\alpha$.
\end{abstract}

\maketitle

\section{Introduction}

\noindent
Logical convergence laws have been proved, or disproved, in many different contexts since the seminal work of 
Glebskii et. al. \cite{Gle} in the late 1960'ies and, independently, Fagin \cite{Fag}.
The results have considered various kinds of structures (e.g. graphs, trees, relational structures, permutations),
various logics (e.g. first-order logic and extensions of it, many valued logics) and various probability distributions
(typically on a set of structures with a common finite domain).
A far from complete list, only meant to illustrate some of the diversity in the results, includes the following studies:
\cite{ABFN, Bal, Bur, HK, Hill, Kai, KL, KPR, KoVa, Lyn, McC, MT, SS}.
The notion of  ``almost sure (or asymptotic)'' elimination of quantifiers or of aggregation functions 
is closely related to the notion of logical convergence law and indeed in several cases a logical convergence law is obtained
as a consequence of almost sure (or asymptotic) elimination of quantifiers or of aggregation functions.

With few exceptions (mentioned below), 
the studies of logical convergence laws have considered ``ordinary'' structures in the sense that for every
relation symbol $R$, say of arity $\nu$, and any choice of elements $a_1, \ldots, a_\nu$ from the domain of a structure,
either $R(a_1, \ldots, a_\nu)$ is true (has value 1) or false (has value 0).
However, many phenomena are naturally described by continuously varying magnitudes.
Therefore much of the methodology in data science, machine learning and artificial intelligence considers continuously varying data and
continuous random variables.
For example, neural networks typically have continuous real valued inputs and outputs, 
and the output often comes together with a probability distribution.

This motivates making investigations about convergence laws for formal logics that ``talk about'' 
real valued relations
(from the logician's point of view) 
or random variables (from the probability theorist's point of view)
where each relation (random variable) comes together with a probability distribution.
To deal with real valued relations we will use the concept of continuous structure that has been studied
in continuous model theory for quite some time; see for example \cite{CK} or \cite[Sections 2--4]{BenY}.
However, continuous model theory focuses on continuous structures with infinite domains
because the main applications are in analysis and algebra, and here we consider finite domains.

Given a finite signature of relation symbols $\sigma$ and a finite domain, which can as well be $[n] := \{1, \ldots, n\}$ 
for some positive integer $n$, we can now let
$\mbW_n$ be the (uncountably infinite) set of all continuous $\sigma$-structures (or ``possible worlds'') with domain $[n]$
such that all relation symbols take values in the unit interval $[0, 1]$.
To each relation symbol $R \in \sigma$, of arity $\nu_R$ say, and tuple $\bar{a} \in [n]^{\nu_R}$ we can associate
a random variable $X_n^{R, \bar{a}} : \mbW_n \to [0, 1]$ such that for $\mcA \in \mbW_n$, 
$X_n^{R, \bar{a}}(\mcA) = R^{\mcA}(\bar{a})$, where $R^{\mcA}(\bar{a})$ is the value of $R(\bar{a})$ in $\mcA$.
To each $X_n^{R, \bar{a}}$ we will associate a continuous probability density function
$\mu_R^{ic} : [0, 1] \to [0, \infty)$ that depends only on $R$ and the identity constraints that are satisfied by $\bar{a}$,
represented by the notation $ic$, and such that $\mu_R^{ic}$ describes the distribution of the values of $X_n^{R, \bar{a}}$.
We will then consider the product space of all $X_n^{R, \bar{a}}$ as $R$ ranges over $\sigma$ and $\bar{a}$ ranges over
$[n]^{\nu_R}$. We let $\mbbP_n$ denote the associated probability distribution on $\mbW_n$, which will have 
the property that for all
$R \in \sigma$ and $\bar{a} \in [n]^{\nu_R}$, $X_n^{R, \bar{a}}$ is independent of all $X_n^{Q, \bar{a}}$ where
$Q \neq R$ or $\bar{a} \neq \bar{b}$.

The logic that will be used, which will be called {\em Continuous Logic with Aggregation functions}, or $CLA$ for short, 
is a many-valued logic with (truth) values in $[0, 1]$.
$CLA$ allows any continuous function $f : [0, 1]^k \to [0, 1]$ as a $k$-ary connective and it allows the use of any
aggregation function $F : \bigcup_{k=1}^\infty [0, 1]^k \to [0, 1]$ that is continuous according to 
Definition~\ref{definition of continuous aggregation function}.
In this sense the aggregation functions maximum, minimum, and average are continuous, so $CLA$ subsumes
first-order logic on continuous structures with finite domain.

\subsection*{Challenges and results, informally}

Let $\varphi(\bar{x})$ be a formula of $CLA$ where $\bar{x} = (x_1, \ldots, x_k)$ is a sequence of distinct free (object) variables.
Also let $n \in \mbbN^+$ (the set of positive integers), let $\bar{a} \in [n]^k$ and let $I \subseteq [0, 1]$.
The question now is how to compute, or estimate, preferably in an efficient way,
the probability that the value of $\varphi(\bar{a})$ in a random $\mcA \in \mbW_n$ is in $I$, and if this
probability converges as  $n$ tends to infinity. 
Recall that the space $\mbW_n$ of continuous $\sigma$-structures with domain $[n]$ is (uncountably) infinite and each
particular member of $\mbW_n$ has probability 0.

We will see that if a formula $\psi(\bar{x})$ (from $CLA$) does not use any aggregation function, then,
for every $\bar{a} \in [n]^{|\bar{x}|}$,
the probability that $\psi(\bar{a})$ belongs to $I$ is equal to an
integral that depends only on $\psi(\bar{x})$, $I$ and the distributions $\mu_R^{ic}$ as $R$ ranges over $\sigma$
and $ic$ ranges over the possible identity constraints that can be formed with $\nu_R$ variables.
In particular, the probability does not depend on the domain size.

Our main result (Theorem~\ref{main result}) is that 
if $\varphi(\bar{x})$ is a formula in $CLA$ and $ic(\bar{x})$ is an identity constraint,
then there is a formula $\psi(\bar{x})$ without any aggregation function, such that as the domain size $n$ tends to infinity,
the probability that $\varphi(\bar{a})$ and $\psi(\bar{a})$ have approximately the same values for all $\bar{a} \in [n]^k$
that satisfy $ic(\bar{x})$ tends to 1.
The proof shows how to ``eliminate'' aggregation functions from $\varphi(\bar{x})$, one by one, 
until one gets a formula $\psi(\bar{x})$ without any
aggregation function which ``almost surely''  gives ``almost the same'' value.
And as said above, the probability that the value of 
$\psi(\bar{a})$ belongs to an interval $I$ can be computed without reference to the domain size.
Thus we have a procedure for estimating the probability that the value of $\varphi(\bar{a})$ belongs to an interval $I$ 
which depends only on $\varphi(\bar{x})$, $I$ and the distributions of the form $\mu_R^{ic}$;
in particular, we have a convergence law for $CLA$.

\subsection*{A comparison with Probabilistic Logic Programming}

To see similarities and differences between the set-up of this article and a well-studied concept we make an informal comparison to 
{\em Probabilistic Logic Programming} (PLP) \cite{RS}, a formalism used in for example 
{\em Statistical Relational Artificial Intelligence}, a subfield of Artificial Intelligence (AI) and Machine Learning
which combines the logical and probabilistic approaches to AI; see e.g. \cite{BKNP, KMG, RKNP}.

Suppose that $P, Q$ and $R$ are relation symbols of arity 1.
A probabilistic logic program (PL-program) consists partly of a set of ``probabilistic facts'' which assign probabilities to some 
relations (atomic formulas).
The probabilistic facts could, for example, say that $P(x)$ and $Q(x)$
have probabilities $\alpha$ and $\beta$, respectively.
This is interpreted as meaning that for every grounding of $P(x)$, say $P(a)$ where $a$ is a member of some domain,
the probability that $P(a)$ holds is $\alpha$, independently of whether or not other ground atoms involving $P$ or $Q$ are true.
The other part of a PL-program is a set of rules, for example $(P(x) \wedge Q(x)) \rightarrow R(x)$, which is implicitly understood to
be universally quantified. 
So if the meaning of $P(x)$, $Q(x)$, and $R(x)$ is ``$x$ is wealthy'', ``$x$ is beautiful'', and
``$x$ is influencial'', respectively, 
then the program expresses that the probability that (a random) $x$ is influencial is (at least) $\alpha\cdot \beta$.

Rather than to consider wealth, beauty, and influentiality as binary properties we can consider the degree, on a continuous
scale, of wealth, beauty, and influence. We normalize the magnitudes so that for example the degree of wealth is always a
number in the unit interval. Suppose that continuous probability density functions
$\mu_P, \mu_Q : [0, 1] \to [0, \infty)$ are associated to $P$ and $Q$, respectively. 
We interpret this as meaning that for any grounding $P(a)$ of $P(x)$ the value of $P(a)$ is distributed
according to $\mu_P$, independently of the values of other ground atoms involving $P$ or $Q$.
Suppose that we have found a continuous function $\msfC : [0, 1]^2 \to [0, 1]$ such that the composition
$g(x) := \msfC(P(x), Q(x))$ gives a (more or less) reliable measure of the degree of influence of a random $x$. 
This means that $\msfC$ outputs the influence of an individual $x$ if it gets  the degrees of wealth and beauty of $x$ as input.
The expression $\msfC(P(x), Q(x))$ is a formula of $CLA$
which can be seen as a ``continuous rule'' which outputs the degree of influence of $x$.
In this context we do not need
an explicit symbol $R$ for ``influence''.

To complicate things, whether $x$ is influencial may not only depend on properties
of $x$ alone (wealth, beauty) but also on degrees of relationships that $x$ has with others.
Suppose that $E(x, y)$ expresses the degree to which $y$ is exposed to $x$ (e.g. in media) and that $f_E : [0, 1] \to [0, \infty)$
is a continuous probability density function, interpreted as saying that for every grounding $E(a, b)$,
where $a$ and $b$ are different, the degree to which $b$ is exposed to
$a$ is distributed according to $f_E$. 
Suppose that the ``degree of influence of $a$'' can be estimated (in a ``continuous way'') by knowing only
the degrees of wealth and beauty of $a$, and the average of the sequence of values of $E(a, b)$ as $b$ ranges over
other individuals of the domain.
With $CLA$ we can then express ``the degree of influence of $x$'' by a 
a formula $\varphi(x)$ which uses average to aggregate over the domain.

\subsection*{Related work}

It was mentioned above that almost all work on logical convergence laws have considered ``conventional'' structures
with true/false (or 1/0) valued relations.
One exception is \cite{Gra} by Grädel et. al. where
the authors prove logical convergence laws for first-order logic with semiring semantics, 
which partly means that grounded atomic formulas
have values in a semiring and for every grounded atomic formula either itself or its negation must have the value 0.
Another exception is \cite{Gol} by Goldbring et. al. who prove an ``approximate 0-1 law'' 
for finite metric spaces in which nontrivial distances are in the interval $[\frac{1}{2}, 1]$.
These metric spaces can be viewed as relational structures with a relation symbol of arity $2$ that takes (truth) values
in $[\frac{1}{2}, 1]$ for unordered pairs of distinct elements.
Finally, the featured graph neural networks considered by Adam-Day et. al. in \cite{Adam-Day1} and \cite{Adam-Day2} 
can be viewed as graphs expanded by some unary relation symbols that can take any nonnegative real values, because
the so-called node features can also be modeled by unary relation symbols that take nonnegative real values.

The last two references have a clear connection to machine learning and AI.
Other studies of logical convergence phenomena in contexts where the formal logic or the probability distribution
is inspired by concepts or methods in machine learning and AI include 
\cite{Jae98a, Jae98b, CM, Kop20, Wei21, KW1, KW2, Wei24, KT, Kop26}, in order of publication year.

\subsection*{Notation and terminology}

\noindent
We let $\mbbN$ denote the set of all nonnegative integers and $\mbbN^+$ the set of all positive integers.
For all $n \in \mbbN^+$, $[n] := \{1, \ldots, n\}$.
Logical (object) variables are denoted by $x$, $y$ or $z$, sometimes with subscripts, and finite sequences 
 of logical variables  are denoted by $\bar{x}$. 
The length of $\bar{x}$ is denoted by $|\bar{x}|$.
Real numbers will be denoted by $p, q$ or $r$, sometimes with subscripts, and finite sequences of reals by $\bar{p}$, $\bar{q}$ or
$\bar{r}$.
For a set $S$, $|S|$ denotes its cardinality.
Elements from some domain of a (continuous) structure are denoted by $a, b, c$, possibly by subscripts, and $\bar{a}$, $\bar{b}$ etc
denote finite sequences of such elements. By $\rng(\bar{a})$ we denote the set of all elements in $\bar{a}$.
If $\bar{a} = (a_1, \ldots, a_k)$ then $\bar{a}b$ denotes the sequence $(a_1, \ldots, a_k, b)$.

{\em For the rest of the article $\sigma$ denotes a nonempty finite set of relation symbols (in the sense of first-order logic) which 
we call a {\bf \em signature}.
We assume that each $R \in \sigma$ has {\bf \em arity} $\nu_R$.}

\section{Continuous Structures and Continuous Logic with Aggregation Functions}

\noindent
In this section we define the logical concepts that we will work with: continuous $\sigma$-structures, 
continuous aggregation functions, and continuous logic with aggregation functions.
We also prove some basic results about the concepts.

\begin{defin}{\rm
A {\bf \em continuous $\sigma$-structure} $\mcA$ is determined by the following information:
\begin{enumerate}
\item A nonempty set $A$, which we call the {\bf \em domain} of $\mcA$.
\item For every $R \in \sigma$ (with arity $\nu_R$), a function $R^\mcA : A^{\nu_R} \to [0, 1]$
(so $R^\mcA(\bar{a}) \in [0, 1]$ for all $\bar{a} \in A^{\nu_R}$).
\end{enumerate}
A continuous $\sigma$-structure is called {\bf \em finite} it its domain is finite.
}\end{defin}

\begin{defin}{\rm
(a) Let $[0, 1]^{<\omega}$ denote the set of all finite nonempty sequences of reals from the unit interval $[0, 1]$, or equivalently,
$[0, 1]^{<\omega} = \bigcup_{k = 1}^\infty [0, 1]^k$.\\
(b) A function $F : [0, 1]^{<\omega} \to [0, 1]$ will be called an {\bf \em aggregation function} if it is symmetric in the sense that
if $(r_1, \ldots, r_k) \in [0, 1]^k$ and $(q_1, \ldots, q_k)$ is a permutation (reordering) of $(r_1, \ldots, r_k)$ then
$F(r_1, \ldots, r_k) = F(q_1, \ldots, q_k)$.
}\end{defin}

\begin{exam}\label{examples of aggregation functions} {\rm 
For $\bar{p} = (p_1, \ldots, p_n) \in [0, 1]^{<\omega}$, define
\begin{enumerate}
\item $\max(\bar{p})$ ($\min(\bar{p})$) to be the {\em maximum} ({\em minimum}) of $p_1, \ldots, p_n$, and
\item $\mr{am}(\bar{p}) = (p_1 + \ldots + p_n)/n$.
\end{enumerate}
}\end{exam}

\noindent
Thus `am' is the {\em arithmetic mean}, or {\em average}.
Definition~\ref{definition of continuous aggregation function} below defines 
the notion of ``continuity'' for aggregation functions.
For an aggregation function $F : [0, 1]^{<\omega} \to [0, 1]$,
the intuition behind the technical part~(2) of 
Definition~\ref{definition of continuous aggregation function} is that,
for every $\varepsilon > 0$, 
if $(r_1, \ldots, r_n), (q_1, \ldots, q_m) \in [0, 1]^{<\omega}$ are sufficiently long sequences
that are sufficiently similar regarding the proportion of entries in sufficently small subintervals of $[0, 1]$ and if none of the
sequences have too large gaps (if the entries are ordered from smaller to larger), 
then $|F(r_1, \ldots, r_n) - F(q_1, \ldots, q_m)| \leq \varepsilon$.

\begin{defin}\label{definition of continuous aggregation function}{\rm
Let $F : [0, 1]^{<\omega} \to [0, 1]$ be an aggregation function.
We call $F$ {\bf \em continuous} if the following two conditions hold:
\begin{enumerate}
\item For every $\varepsilon > 0$ there is $\delta > 0$ such that for all $n$,
if $(q_1, \ldots, q_n)$, $(q'_1, \ldots, q'_n)$ $\in [0, 1]^n$ and $|q_i - q'_i| \leq \delta$ for all $i = 1, \dots, n$,
then $|F(q_1, \ldots, q_n) - F(q'_1, \ldots, q'_n)| \leq \varepsilon$.

\item For every $\varepsilon > 0$ there are $\delta > 0$ and $M, N \in \mbbN^+$
such that if $\alpha_0, \ldots, \alpha_{M-1} \in [0, 1]$ and 
$(q_1, \ldots, q_n), (q'_1, \ldots, q'_m) \in [0, 1]^{<\omega}$ (where we may have $n \neq m$) are such that
conditions (a)--(d) below hold, then $|F(q_1, \ldots, q_n) - F(q'_1, \ldots, q'_m)| \leq \varepsilon$:
\begin{enumerate}
\item $n, m \geq N$, 

\item for all $i = 0, \ldots, M-1$, if $\alpha_i > 0$ then $\alpha_i > \delta$,
and (regardless of whether $\alpha_i > 0$ or not)
\begin{align*}
&\frac{\big|\big\{l : q_l \in \big[\frac{i}{M}, \frac{i+1}{M}\big] \big\}\big|}{n} \in (\alpha_i - \delta, \alpha_i + \delta), \\
&\frac{\big|\big\{l : q'_l \in \big[\frac{i}{M}, \frac{i+1}{M}\big] \big\}\big|}{m} \in (\alpha_i - \delta, \alpha_i + \delta), 
\end{align*}

\item  if $i < j < k$ or $i > j > k$, and in addition $\alpha_i > 0$ and $\alpha_j =  0$, then $\alpha_k = 0$, and 

\item if $i < j < k$ or $i > j > k$, and in addition $\alpha_i = \alpha_j = 0$ and $\alpha_k > 0$, then 
\begin{align*}
\{l : q_l \in \big[ {\scriptstyle \frac{i}{M}, \frac{i+1}{M} } \big] \big\} = 
\{l : q'_l \in \big[ {\scriptstyle\frac{i}{M}, \frac{i+1}{M} } \big] \big\} = \es, 
\end{align*}
\end{enumerate}
\end{enumerate}

}\end{defin}

\begin{lem}
The aggregation functions maximum, minimum and arithmetic mean are continuous.
\end{lem}

\noindent
{\bf Proof.}
We start with maximum.
It is immediately clear that max satisfies condition~(1) of Definition~\ref{definition of continuous aggregation function}
so we only verify that~(2) of the same definition holds.
Let $\varepsilon > 0$.
Suppose that 
conditions (a)--(d) of part~(2) of 
Definition~\ref{definition of continuous aggregation function}
are satisfied by $\delta > 0$, $M, N \in \mbbN^+$,
$\alpha_0, \ldots, \alpha_{M-1} \in [0, 1]$, and 
$(q_1, \ldots, q_n), (q'_1, \ldots, q'_m)$.
For $i = 0, \ldots, M-1$, let $P_i = \big\{l : q_l \in \big[\frac{i}{M}, \frac{i+1}{M}\big] \big\}$
and  $P'_i = \big\{l : q'_l \in \big[\frac{i}{M}, \frac{i+1}{M}\big] \big\}$.
Let $s$ be maximal such that $P_s \neq \es$.
Then $s/M \leq \max(\bar{q}) \leq (s+1)/M$.

Suppose for a contradiction that $i > s$ and $\alpha_i > 0$.
By assumption (2)(a) (of Definition~\ref{definition of continuous aggregation function}) we
get $\alpha_i > \delta$, so $P_i \neq \es$ by assumption (2)(b) and this contradicts the choice of $s$.
We conclude that $\alpha_i = 0$ for all $i = s+1, \ldots, M-1$.

Assuming that $\delta > 0$ is small enough we have $\alpha_0 + \ldots + \alpha_{M-1} \approx 1$ so there is 
$i \leq s$ such that $\alpha_i > 0$. Let $t$ be maximal such that $\alpha_t > 0$; hence $t \leq s$.
By~(2)(d) we have $P'_i = \es$ for all $i > s+1$.
Hence $\max(\bar{q}') \leq (s+2)/M$.

Suppose, for a contradiction, that $P'_{s-1} = \es$.
Then $\alpha_{s-1} = 0$,
because $\alpha_{s-1} > 0$ would imply $\alpha_{s-1} > \delta$ and hence $P'_{s-1} \neq \es$.
If $\alpha_s > 0$ then  (2)(c) implies that $\alpha_i = 0$ for all $i < s$ which contradicts that $t \leq s$
and $\alpha_t > 0$.
Thus $\alpha_{s-1} = \alpha_s = 0$, $\alpha_t > 0$, and $t < s-1 < s$, so (2)(d) implies
that $P_s = \es$ which contradicts the choice of $s$.

Hence we conclude that $P'_{s-1} \neq \es$ and we get
$(s-1)/M \leq \max(\bar{q}') \leq (s+2)/M$.
It follows that $|\max(\bar{q}) - \max(\bar{q}')| \leq 2/M$.
So if $\delta > 0$ is chosen small enough, $M$ and $N$ large enough,
and conditions~(2)(a)--(2)(d) hold, then 
 $|\max(\bar{q}) - \max(\bar{q}')| \leq \varepsilon$.

The proof that minimum is continuous is similar, so we next consider the arithmetic mean (average), abbreviated `am'.
It is again clear that am has property~(1) of Definition~\ref{definition of continuous aggregation function},
so we only verify property~(2).
Suppose that $\bar{q} = (q_1, \ldots, q_n)$ and $\bar{q}' = (q'_1, \ldots, q'_m)$ 
satisfy conditions (a)--(d) of part~(2) of 
Definition~\ref{definition of continuous aggregation function}.
Then 
\[
\sum_{i=0}^{M-1} \frac{i}{M}(\alpha_i - \delta) \leq \mr{am}(\bar{q}) \leq 
\sum_{i=0}^{M-1} \frac{i+1}{M}(\alpha_i + \delta)
\]
and similarly if $\bar{q}$ is replaced by $\bar{q}'$.
Hence $|\mr{am}(\bar{q}) - \mr{am}(\bar{q}')| \leq 
\sum_{i=0}^{M-1} \frac{\alpha_i + 2i\delta + \delta}{M}$, and if $\delta$ is small enough 
(relative to $M$) then this is approximately $\sum_{i=0}^{M-1} \frac{\alpha_i}{M}$.
For sufficiently small $\delta$, the sum $\alpha_0 + \ldots + \alpha_{M-1}$ is approximately 1,
so $\sum_{i=0}^{M-1} \frac{\alpha_i}{M} \approx \frac{1}{M}$ which can be made as small as we like if $M$ is chosen
large enough.
\hfill $\square$

\medskip

\noindent
One may note that in the context of true/false-valued relations studied in \cite{KW2} 
the aggregation functions maximum and minimum are not continuous (also called ``strongly admissible'' in \cite{KW2})
according to the definitions used there 
(but only ``semi-continuous'', or ``admissible'').

\begin{defin}\label{syntax of CLA}{\bf (Syntax of $CLA$)} {\rm 
We define $CLA$, the set of formulas of {\bf \em continuous logic with aggregation functions}, as follows,
where for every $\varphi \in CLA$ we simultaneously define the set of its free variables, denoted $Fv(\varphi)$:
\begin{enumerate}
\item  For each $c \in [0, 1]$, $c \in CLA$ and $Fv(c) = \es$. 

\item For all logical variables $x$ and $y$, `$x = y$' belongs to $CLA$ and $Fv(x = y) = \{x, y\}$.

\item For every $R \in \sigma$ (of arity $\nu_R$) and any choice of logical variables $x_1, \ldots, x_{\nu_R}$, $R(x_1, \ldots, x_{\nu_R})$ belongs to 
$CLA$ and  $Fv(R(x_1, \ldots, x_{\nu_R})) = \{x_1, \ldots, x_{\nu_R}\}$.

\item If $k \in \mbbN^+$, $\varphi_1, \ldots, \varphi_k \in CLA$ and
$\msfC : [0, 1]^k \to [0, 1]$ is a continuous function, then the expression
$\msfC(\varphi_1, \ldots, \varphi_k)$ is a formula of $CLA$ and
its set of free variables is $Fv(\varphi_1) \cup \ldots \cup Fv(\varphi_k)$.

\item Suppose that $\varphi \in CLA$, $y$ is a logical variable, and that $F : [0, 1]^{<\omega}  \to [0, 1]$ is a
continuous aggregation function.
Then the expression $F(\varphi : y)$
is a formula of $CLA$ and its set of free variables is $Fv(\varphi) \setminus \{y\}$.
Hence this construction binds the variable $y$.
\end{enumerate}
}\end{defin}

\noindent
When denoting a formula (of $CLA$) by $\varphi(\bar{x})$ (or $\psi(\bar{x})$ etc), 
{\em it will be assumed that all free variables of
$\varphi(\bar{x})$ appear in the sequence $\bar{x}$} (but we do not insist that every variable in $\bar{x}$ appears in 
$\varphi(\bar{x})$).
Formulas of the forms 1, 2 and 3 in Definition~\ref{syntax of CLA} will be called {\bf \em atomic}.

\begin{defin}\label{semantics of CLA}{\bf (Semantics of $CLA$)} {\rm
For every $\varphi \in CLA$ and every sequence of distinct logical variables $\bar{x} = (x_1, \ldots, x_k)$ such that 
$Fv(\varphi) \subseteq \{x_1, \ldots, x_k\}$ we associate a mapping from pairs $(\mcA, \bar{a})$,
where $\mcA$ is a finite continuous $\sigma$-structure and $\bar{a} = (a_1, \ldots, a_k) \in A^k$, to $[0, 1]$.
The number in $[0, 1]$ to which $(\mcA,\bar{a})$ is mapped is denoted $\mcA(\varphi(\bar{a}))$
and is defined by induction on the complexity of formulas, as follows:
\begin{enumerate}
\item If $\varphi(\bar{x})$ is a constant $c$ from $[0, 1]$, then $\mcA(\varphi(\bar{a})) = c$.

\item If $\varphi(\bar{x})$ has the form $x_i = x_j$, then $\mcA(\varphi(\bar{a})) = 1$ if $a_i = a_j$,
and otherwise $\mcA(\varphi(\bar{a})) = 0$.

\item For every $R \in \sigma$, if $\varphi(\bar{x})$ has the form $R(x_{i_1}, \ldots, x_{i_{\nu_R}})$,
then $\mcA(\varphi(\bar{a})) = R^\mcA(a_{i_1}, \ldots, a_{i_{\nu_R}})$.

\item If $\varphi(\bar{x})$ has the form $\msfC(\varphi_1(\bar{x}), \ldots, \varphi_k(\bar{x}))$,
where $\msfC : [0, 1]^k \to [0, 1]$ is continuous, then
\[
\mcA\big(\varphi(\bar{a})\big) \ = \ 
\msfC\big(\mcA(\varphi_1(\bar{a})), \ldots, \mcA(\varphi_k(\bar{a}))\big).
\]

\item Suppose that $\varphi(\bar{x})$ has the form 
$F(\psi(\bar{x}, y) : y)$
where $y$ is a logical variable that does not occur in $\bar{x}$
and $F : [0, 1]^{<\omega} \to [0, 1]$ is a continuous aggregation function.
Then 
\[
\mcA\big(\varphi(\bar{a})\big) = F(r_1, \ldots, r_n)
\]
where, for $i = 1, \ldots, n$, $r_i = \mcA\big(\psi(\bar{a}, b_i)\big)$ and $b_1, \ldots, b_n$ is an enumeration (without repetition)
of all elements in $A \setminus \rng(\bar{a})$ where $A$ is the domain of $\mcA$.
\end{enumerate}
}\end{defin}

\begin{rem}\label{remark on FO}{\rm 
(i) Note that  we have $\mcA(R(\bar{a})) = R^\mcA(\bar{a})$
for every $R \in \sigma$, every continuous $\sigma$-structure $\mcA$, and every $\bar{a} \in A^{\nu_R}$.

(ii) The classical connectives $\neg$, $\wedge$, $\vee$ and $\rightarrow$, which can be seen as functions from
$\{0, 1\}$ or $\{0, 1\}^2$ to $\{0, 1\}$, can be extended, for example by Lukasiewicz semantics \cite{LT}, to continuous functions from
$[0, 1]$ or $[0, 1]^2$ to $[0, 1]$, so these connectives can be expressed with $CLA$.
Likewise, the aggregation functions max and min can be used to express existential and universal quantifications,
so $CLA$ subsumes first-order logic on ``ordinary'' $\sigma$-structures $\mcA$ in which $R^\mcA(\bar{a})$ is either 0 or 1 for all 
$R \in \sigma$ and $\bar{a} \in A^{\nu_R}$.
This may not be immediately clear because according to the semantics of $CLA$, 
the construction $F(\varphi(x_1, \ldots, x_k, y) : y)$ only aggregates only over $y$
that are different from all $x_1, \ldots, x_k$. 
But the first-order construction $\exists y \varphi(x_1, \ldots, x_k, y)$ can be expressed
in $CLA$ by 
\[
\mr{max}_{k+1}\big(\varphi(x_1, \ldots, x_k, x_1), \ldots, \varphi(x_1, \ldots, x_k, x_k), \max(\varphi(x_1, \ldots, x_k, y) : y)\big)
\]
where $\max_{k+1} : [0, 1]^{k+1} \to [0, 1]$ is the continuous function which outputs the largest entry in a $(k+1)$-tuple.
}\end{rem}

\begin{defin}
(a) A formula of $CLA$ is called {\bf \em aggregation-free} if it contains no aggregation function, that is, if
it was constructed by only using parts 1 -- 4 of Definition~\ref{syntax of CLA}.\\
(b) The formulas $\varphi(\bar{x}), \psi(\bar{x}) \in CLA$ are {\bf \em equivalent} if for every finite continuous $\sigma$-structure
$\mcA$ and every $\bar{a} \in A^{|\bar{x}|}$, $\mcA(\varphi(\bar{a})) = \mcA(\psi(\bar{a}))$.
\end{defin}

\noindent
The next results allows us to simplify aggregation-free formulas to such which use exactly one connective
(i.e. continuous function $\msfC : [0, 1]^k \to [0, 1]$ for some $k$), and exactly once.

\begin{lem}\label{simplifications of aggregation-free formulas}
If $\psi(\bar{x}) \in CLA$ is aggregation-free then, for some $m$, it is equivalent to a formula 
$\msfC(\psi_1(\bar{x}), \ldots, \psi_m(\bar{x}))$ where $\msfC : [0, 1]^m \to [0, 1]$ is continuous, for all $i$, $\psi_i(\bar{x})$
is atomic, and if $i \neq j$ then $\psi_i(\bar{x})$ and $\psi_j(\bar{x})$ are different formulas.
\end{lem}

\noindent
{\bf Proof.}
We use induction on the complexity of aggregation-free formulas.
Suppose that $\psi(\bar{x})$ is aggregation-free.
If $\psi(\bar{x})$ is atomic then it is equivalent to $\msfC(\psi(\bar{x}))$ where $\msfC : [0, 1] \to [0, 1]$ is the identity function.

Now suppose that $\psi(\bar{x})$ has the form $\msfD(\psi_1(\bar{x}), \ldots, \psi_t(\bar{x}))$ where
$\msfD : [0, 1]^t \to [0, 1]$ is continuous. 
By the induction hypothesis, for every $i = 1, \ldots, t$, $\psi_i(\bar{x})$ is equivalent to a formula of the form
$\msfD_i(\psi_{i, 1}(\bar{x}), \ldots, \psi_{i, s_i}(\bar{x}))$ where each $\psi_{i, j}(\bar{x})$ is atomic
and $\msfD_i : [0, 1]^{s_i} \to [0, 1]$ is continuous.
Now let $m = s_1 + \ldots + s_t$ and let $\msfC : [0, 1]^m \to [0, 1]$ be the composition defined by
\[
\msfC(r_1, \ldots, r_m) = 
\msfD(\msfD_1(r_1, \ldots, r_{s_1}), \ldots, \msfD_t(r_{m - s_t}, \ldots, r_m)),
\]
so $\msfC$ is continuous.
Then $\psi(\bar{x})$ is equivalent to 
\[
\msfC(\psi_{1, 1}(\bar{x}), \ldots, \psi_{1, s_1}(\bar{x}), \ldots, \psi_{t, 1}(\bar{x}), \ldots, \psi_{t, s_t}(\bar{x})).
\]
To simplify notation let us rename the formulas $\psi_{i, j}$ so that $\psi(\bar{x})$ is equivalent
to $\msfC(\psi_1(\bar{x}), \ldots, \psi_m(\bar{x}))$. Suppose that for $1 \leq i < j \leq m$, $\psi_i(\bar{x})$ is the same formula
as $\psi_j(\bar{x})$. For notational simplicity and without loss of generality, suppose that $j = m$.
Then define $\msfC' : [0, 1]^{m-1} \to [0, 1]$ by $\msfC'(r_1, \ldots, r_{m-1}) = \msfC(r_1, \ldots, r_{m-1}, r_{m-1})$.
Then $\msfC'$ is continuous and $\msfC'(\psi_1(\bar{x}), \ldots, \psi_{m-1}(\bar{x}))$ is equivalent to 
$\msfC(\psi_1(\bar{x}), \ldots, \psi_m(\bar{x}))$.
By continuing in this way as long as necessary we can make sure to get a formula like
$\msfC(\psi_1(\bar{x}), \ldots, \psi_m(\bar{x}))$ that satisfies the additional condition that $\psi_i(\bar{x})$ 
is different from $\psi_j(\bar{x})$ if $i \neq j$.
\hfill $\square$

\medskip

\noindent
The idea with an identity constraint, that we define next, 
is that it specifies which pairs of elements in a tuple/sequence are equal and which are not.

\begin{defin}\label{definition of identity constraint}{\rm
Let $\bar{x} = (x_1, \ldots, x_k)$ be a sequence of (not necessarily distinct) variables.\\
(i) Let $K = \{(i, j) : 1 \leq i < j \leq k\}$. An {\bf \em identity constraint for $\bar{x}$} is a first-order formula of the following form for 
some $I \subseteq K$: 
$\bigwedge_{(i, j) \in I} x_i = x_j \ \wedge \ \bigwedge_{(i, j) \in K \setminus I} x_i \neq x_j$.\\
(ii) Let $ic(\bar{x})$ be an identity constraint for $\bar{x}$. 
Suppose that $\bar{a} = (a_1, \ldots, a_k)$ where $a_1, \ldots, a_k$ may be any objects.
We say that $\bar{a}$ {\bf \em satifies} $ic(\bar{x})$ if for all $1 \leq i < j \leq k$,
$a_i = a_j$ if and only if $ic(\bar{x}) \models x_i = x_j$, where `$\models$' has the usual meaning for first-order logic.\\
(iii) Suppose that $ic(\bar{x})$ is an identity constraint for $\bar{x}$ and let $\bar{x}'$ be a sequence of variables such that
$\rng(\bar{x}') \subseteq \rng(\bar{x})$. The {\bf \em restriction of $ic(\bar{x})$ to $\bar{x}'$} is an identity constraint $ic'(\bar{x}')$
for $\bar{x}'$ such that $ic(\bar{x}) \wedge ic'(\bar{x}')$ is consistent (or equivalently, such that $ic'(\bar{x}')$
is a logical consequence of $ic(\bar{x})$). 
Note that, syntactically, the restriction of $ic(\bar{x})$ to $\bar{x}'$ is not unique, but it is unique up to logical equivalence,
so as {\em the} restriction we can just choose one of the logically equivalent candidates.
}\end{defin}

\noindent
{\em Observe that an identity constraint $ic(\bar{x})$ may be inconsistent (as it may violate the transitivity of `$=$',
or the reflexivity of `$=$' if $x_i$ and $x_j$ are the same variable for some $i \neq j$)
but every result that will follow and mentions an identity constraint is vacuously true if the identity constraint is inconsistent.}
To allow repetitions in the sequence $\bar{x}$ may seem odd at this point but it will be technically convenient later.
Informally, the next lemma says that for every aggregation-free formula $\varphi(\bar{x})$ and identity constraint $ic(\bar{x})$
there is an aggregation-free formula $\psi(\bar{x})$ which uses only atomic formulas of the form (3) in 
Definition~\ref{syntax of CLA} and such that, conditioned on $\bar{a}$ satisfying $ic(\bar{x})$, $\varphi(\bar{a})$ and $\psi(\bar{a})$
have he same value.

\begin{lem}\label{reduction to simpler formula}
Let $\bar{x} = (x_1, \ldots, x_k)$ be a sequence of distinct variables, 
let $ic(\bar{x})$ be an identity constraint for $\bar{x}$, 
and let $\varphi(\bar{x})$ be an aggregation-free formula.
Then there is an aggregation-free formula $\psi(\bar{x})$ such that 
\begin{enumerate}
\item if $\mcA$ is a finite continuous $\sigma$-structure and $\bar{a} \in A^k$ satisfies $ic(\bar{x})$, then
$\mcA(\varphi(\bar{a})) = \mcA(\psi(\bar{a}))$, and
\item $\psi(\bar{x})$ is either a constant 
(i.e. of the form (1) in Definition~\ref{syntax of CLA})
or $\psi(\bar{x})$ has the form
$\msfC(R_1(\bar{x}_1), \ldots, R_m(\bar{x}_m))$
where $\msfC : [0, 1]^m \to [0, 1]$ is continuous, for each 
$i = 1, \ldots, m$, 
$R_i \in \sigma$, $\rng(\bar{x}_i) \subseteq \rng(\bar{x})$, $|\bar{x}_i| = \nu_{R_i}$, and if
$i \neq j$ then $R_i \neq R_j$ or $\bar{x}_i \neq \bar{x}_j$.
\end{enumerate}
\end{lem}

\noindent
{\bf Proof.}
Let $\bar{x} = (x_1, \ldots, x_k)$ be a sequence of distinct variables,
let $ic(\bar{x})$ be an identity constraint for $\bar{x}$,
and let $\varphi(\bar{x})$ be an aggregation-free formula.
By Lemma~\ref{simplifications of aggregation-free formulas} we may assume that 
$\varphi(\bar{x})$ has the form 
\[
\msfD(c_1, \ldots, c_t, \theta_{t+1}(\bar{x}_{t+1}), \ldots, \theta_{t+s}(\bar{x}_{t+s}),
R_{t+s+1}(\bar{x}_{t+s+1}), \ldots, R_{t+s+u}(\bar{x}_{t+s+u}))
\]
where
$\msfD : [0, 1]^{t+s+u} \to [0, 1]$ is continuous, each $c_i$ is a constant, each $\theta_i(\bar{x}_i)$ has the form $x_j = x_l$ for some
distinct $j, l \in \{1, \ldots, k\}$ (if $j = l$ we can replace $x_j = x_l$ by 1), $R_i \in \sigma$,
$\rng(\bar{x}_i) \subseteq \rng(\bar{x})$, and $|\bar{x}_i| = \nu_{R_i}$.
Moreover, by the same lemma, we may assume that if $t + s < i < j \leq t + s + u$ then $R_i \neq R_j$ or $\bar{x}_i \neq \bar{x}_j$.

Let $t+1 \leq i \leq t+s$ and suppose that $\theta_i$ is the formula $x_j = x_l$ where $j \neq l$.
Then let $c_i = 1$ if $ic(\bar{x}) \models x_j = x_l$ and otherwise let $c_i = 0$.
It follows that if $u = 0$ then
$\varphi(\bar{x})$ is $\msfD(c_1, \ldots, c_t, \theta_{t+1}(\bar{x}_{t+1}), \ldots, \theta_{t+s}(\bar{x}_{t+s})$
and if $c := \msfD(c_1, \ldots, c_{t+s})$, then for every finite continuous $\sigma$-structure $\mcA$
and every $\bar{a} \in A^k$ that satisfies $ic(\bar{x})$
we have $\mcA(\varphi(\bar{a})) = c = \mcA(c)$.

Now suppose that $u > 0$ and define $\msfC : [0, 1]^u \to [0, 1]$ by 
\[
\msfC(r_1, \ldots, r_u) = \msfD(c_1, \ldots, c_{t+s}, r_1, \ldots, r_u).
\]
Then $\msfC$ is continuous and 
for all finite continuous $\sigma$-structures $\mcA$ and all $\bar{a} \in [n]^k$ that satisfy $ic(\bar{x})$,
$\mcA(\varphi(\bar{a})) = \mcA(\msfC(R_{t+s+1}(\bar{a}_{t+s+1}), \ldots, R_{t+s+u}(\bar{a}_{t+s+u})))$
where $\bar{a}_i$ is the subsequence of $\bar{a}$ that corresponds to $\bar{x}_i$.
\hfill $\square$

\section{Probability Distributions}

\noindent
In this section we define the probability theoretic concepts of this study.

\begin{defin}{\rm
Let $n \in \mbbN^+$.\\
(a) $\mbW_n$ denotes the set of all continuous $\sigma$-structures with
domain $[n] = \{1, \ldots, n\}$. \\
(b) $\xi_n := \prod_{R \in \sigma} n^{\nu_R}$.\\
(c) $\Xi_n := \big\{(R, \bar{a}) : R \in \sigma \text{ and } \bar{a} \in [n]^{\nu_R} \big\}$
(so $\xi_n = |\Xi_n|$).\\
(d) Let $\prec_n$ be a linear order on $\Xi_n$ such that if $n < m$, then $\prec_n$ is the
restriction of $\prec_m$ to $\Xi_n$. \\
(e) Define $\mfs_n : \mbW_n \to [0, 1]^{\xi_n}$ by, for all $\mcA \in \mbW_n$, letting
$\mfs_n(\mcA) = (r_1, \ldots, r_{\xi_n})$ where, for all $i = 1, \ldots, \xi_n$,
$r_i = R^\mcA(\bar{a})$ where $(R, \bar{a})$ is the $i^{th}$ element 
of $\Xi_n$ with respect to the order $\prec_n$.
}\end{defin}

\noindent
From the above definition and the definition of continuous $\sigma$-structure we immediately get:

\begin{lem}\label{identification of W-n with an interval}
For all $n \in \mbbN^+$, $\mfs_n : \mbW_n \to [0, 1]^{\xi_n}$ is bijective.
\end{lem}

\noindent
Thus we can identify $\mbW_n$ with $[0, 1]^{\xi_n}$ where the latter is a compact metric space.
The next couple of definitions show how we can view $\mbW_n$ as a probability space.
The intuition is that for each $R \in \sigma$ and each identity constraint $ic(x_1, \ldots, x_{\nu_R})$,
if $(a_1, \ldots, a_{\nu_R}) \in [n]^{\nu_R}$ satisfies $ic(x_1, \ldots, x_{\nu_R})$ then 
the distribution of 
the random variable $X_n^{R, \bar{a}} : \mbW_n \to [0, 1]$ 
defined by $X_n^{R, \bar{a}}(\mcA) = R^{\mcA}(\bar{a})$ for all $\mcA \in \mbW_n$
is given by a continuous probability density function $\mu_R^{ic} : [0, 1] \to [0, \infty)$ which depends only
on $R$ and $ic$, where we can freely choose $\mu_R^{ic}$ as long as it is continuous and
$\int_{[0, 1]} \mu_R^{ic}(x) dx = 1$. So if $I \subseteq [0, 1]$ is (Lebesgue) measurable then
the probability that $X_n^{R, \bar{a}}$ belongs to $I$ equals
$\int_I \mu_R^{ic}(x) dx$. Moreover, we want this probability to be independent of the value of $X_n^{Q, \bar{b}}$ 
if $Q \neq R$ or if $\bar{b} \neq \bar{a}$.

\begin{defin}\label{definition of mu-R etc}
{\rm  Let $n \in \mbbN^+$.\\
(a) To every pair $(R, ic)$, where $R \in \sigma$ and $ic(x_1, \ldots, x_{\nu_R})$ is an
identity constraint for $(x_1, \ldots, x_{\nu_R})$,
we associate a continuous probability density function
(also called probability distribution) $\mu_R^{ic} : [0, 1] \to [0, \infty)$.\\
(b) For all $\mcA \in \mbW_n$ define 
\[
\pi_n(\mcA) = \prod_{R \in \sigma} \ \ 
\prod_{\substack{ic \text{ is an } \\ \text{identity constraint}\\ \text{for } (x_1, \ldots, x_{\nu_R})}} \ 
\prod_{\substack{\bar{a} \in [n]^{\nu_R} \text{ and}\\ \bar{a} \text{ satisfies } ic}} \mu_R^{ic}(R^\mcA(\bar{a})).
\]
(c) For every $i = 1, \ldots, \xi_n$ we define $\lambda_i : [0, 1] \to [0, \infty)$ by choosing
the $i^{th}$ pair $(R, \bar{a})$ in $\Xi_n$ according to the order $\prec_n$ and, for every $r \in [0, 1]$,
letting $\lambda_i(r) = \mu_R^{ic}(r)$ where $ic$ is the identity constraint for sequences of length $\nu_R$ that
$\bar{a}$ satisfies.\\
(d) For every $(r_1, \ldots, r_{\xi_n}) \in [0, 1]^{\xi_n}$ define 
$\pi_n^*(r_1, \ldots, r_{\xi_n}) = \lambda_1(r_1) \cdot \ldots \cdot \lambda_{\xi_n}(r_{\xi_n})$.
}\end{defin}

\begin{lem}\label{lambda-i and pi* are distributions}
Let $n \in \mbbN^+$.\\
(a) For each $i \in [\xi_n]$, $\lambda_i$ is continuous and $\int_{[0, 1]} \lambda_i(r) dr = 1$.\\
(b) $\pi_n^*$ is a continuous function from $[0, 1]^{\xi_n}$ to $[0, 1]$ and 
\[
\int_{[0, 1]^{\xi_n}} \pi_n^*(r_1, \ldots, r_{\xi_n}) d r_1 \ldots d r_{\xi_n} = 1,
\] 
so $\pi_n^*$ is a probability density function on
$[0, 1]^{\xi_n}$.\\
(c) For all $\mcA \in \mbW_n$,
$\pi_n(\mcA) = \pi_n^*(\mfs_n(\mcA))$.
\end{lem}

\begin{proof}
(a) For each $i \in [\xi_n]$ there are $R \in \sigma$ and an identity constraint $ic$
such that $\lambda_i = \mu_R^{ic}$ where $\mu_R^{ic} : [0, 1] \to [0, \infty)$ is a probability density function and hence $\int_{[0, 1]} \mu_R^{ic}(r) dr = 1$.

(b) We have 
\begin{align*}
&\int_{[0, 1]^{\xi_n}} \pi_n^*(r_1, \ldots, r_{\xi_n}) d r_1 \ldots d r_{\xi_n} = 
\int_{[0, 1]^{\xi_n}} \lambda_1(r_1) \cdot \ldots \cdot \lambda_{\xi_n}(r_{\xi_n})  d r_1 \ldots d r_{\xi_n} = \\
&\bigg( \int_{[0, 1]} \lambda_1(r_1) d r_1 \bigg) \cdot \ldots \cdot 
\bigg( \int_{[0, 1]} \lambda_{\xi_n}(r_{\xi_n}) d r_{\xi_n} \bigg) = 1 \cdot \ldots \cdot 1  = 1.
\end{align*}
Part (c) follows directly from the definitions.
\end{proof}

\begin{defin}\label{definition of distribution}{\rm
Let $n \in \mbbN^+$.\\
(a) For every (Lebesgue) measurable $X \subseteq [0, 1]^{\xi_n}$, let
\[
\mbbP_n^*(X) =\int_X \pi_n^*(x_1, \ldots, x_{\xi_n}) d x_1 \ldots d x_{\xi_n}.
\]
(b) We call $\mbX \subseteq \mbW_n$ {\bf \em measurable} if $\mfs_n(\mbX) := \{\mfs_n(\mcA) : \mcA \in \mbX\}$
is a measurable subset of $[0, 1]^{\xi_n}$.\\
(c) For every measurable $\mbX \subseteq \mbW_n$ define $\mbbP_n(\mbX) = \mbbP_n^*(\mfs_n(\mbX))$.
}\end{defin}

\noindent
We now make sure that certain natural subsets of $\mbW_n$
are measurable.

\begin{lem}\label{definable sets are measurable}
Let $n \in \mbbN^+$, $\varphi(\bar{x}) \in CLA$, and $\bar{a} \in [n]^{|\bar{x}|}$.\\
(a) Define $f : [0, 1]^{\xi_n} \to [0, 1]$ by $f(\bar{r}) = \mfs_n^{-1}(\bar{r})(\varphi(\bar{a}))$,
so $f(\bar{r}) = \mcA(\varphi(\bar{a}))$ where $\mcA \in \mbW_n$ is 
such that $\mfs_n(\mcA) = \bar{r}$.
Then $f$ is continuous.\\
(b) If $I \subseteq [0, 1]$ is an interval then $\{\mcA \in \mbW_n : \mcA(\varphi(\bar{a})) \in I\}$
is a measurable subset of $\mbW_n$.
\end{lem}

\begin{proof}
We begin by showing that (b) follows from (a). 
By Definition~\ref{definition of distribution}, to prove (b)
it suffices to prove that $X \subseteq [0, 1]^{\xi_n}$, where 
$X := \{\mfs_n(\mcA) : \mcA \in \mbW_n \text{ and } \mcA(\varphi(\bar{a})) \in I\}$,
is measurable.
By part~(a), $f$ as defined in that part is continuous, hence a measurable function.
Since every interval is measurable and we have
$X = f^{-1}(I)$ it follows that $X$ is measurable.
Hence it remains to prove~(a).

Let $f$ be defined as in part~(a).
We use induction on the complexity of $\varphi(\bar{x})$.
The base case concerns formulas according to cases (1)--(3) of 
Definition~\ref{syntax of CLA}.

Suppose that $\varphi(\bar{x})$ is a constant $c \in [0, 1]$.
Then $f(\bar{r}) = c$ for all $\bar{r} \in [0, 1]^{\xi_n}$ so $f$ is continuous.

Suppose that $\varphi(\bar{x})$ has the form $x_i = x_j$ for some $x_i$ and $x_j$ in $\bar{x}$.
If $x_i$ and $x_j$ are the same variable then $f(\bar{r}) = 1$ for all $\bar{r} \in [0, 1]^{\xi_n}$ so $f$ is continuous.
Now suppose that $x_i$ and $x_j$ are different variables. 
If $a_i = a_j$ (where $a_i$ and $a_j$ are the $i^{th}$ respectively $j^{th}$ entries in $\bar{a}$)
then $f(\bar{r}) = 1$ for  all $\bar{r} \in [0, 1]^{\xi_n}$; otherwise $f(\bar{r}) = 0$  for all $\bar{r} \in [0, 1]^{\xi_n}$.
In either case, $f$ is continuous.

Suppose that $\varphi(\bar{x})$ has the form $R(x_{i_1}, \ldots, x_{i_{\nu_R}})$ for some $R \in \sigma$
and $x_{i_1},  \ldots, x_{i_{\nu_R}}$ from $\bar{x}$.
For some $j \in [\xi_n]$, $(R, (a_{i_1}, \ldots, a_{i_{\nu_R}}))$ is the $j^{th}$ element of $\Xi_n$
according to the order $\prec_n$. 
Then $f(r_1, \ldots, r_{\xi_n}) = r_j$ for all $r_1, \ldots, r_{\xi_n} \in [0, 1]$, so $f$ is continuous.

Now we turn to the inductive case which concerns parts~(4) and (5) of 
Definition~\ref{syntax of CLA}.
Suppose that $\varphi(\bar{x})$ has the form $\msfC(\varphi_1(\bar{x}), \ldots, \varphi_k(\bar{x}))$
where $\msfC : [0, 1]^k \to [0, 1]$ is continuous.
By the induction hypothesis there are, for all $i = 1, \ldots, k$, continuous $f_i : [0, 1]^{\xi_n} \to [0, 1]$ such that 
for all $\mcA \in \mbW_n$, $f_i(\mfs_n(\mcA)) = \mcA(\varphi_i(\bar{a}))$.
Then, for all $\mcA \in \mbW_n$, $\mcA(\varphi(\bar{a})) = \msfC(f_1(\mfs_n(\mcA)), \ldots, f_k(\mfs_n(\mcA)))$.
It follows that for all $\bar{r} \in [0, 1]^{\xi_n}$, 
$f(\bar{r}) = \msfC(f_1(\bar{r}), \ldots, f_k(\bar{r}))$, so $f$ is continuous.

Finally, suppose that $\varphi(\bar{x})$ has the form 
$F(\psi(\bar{x}, y), : y)$ where $F : [0, 1]^{<\omega} \to [0, 1]$ is a continuous aggregation function.
Let $m := n - |\rng(\bar{a})|$ and let $b_1, \ldots, b_m$ enumerate $[n] \setminus \rng(\bar{a})$.
By the induction hypothesis, for each $i = 1, \ldots, m$,
there is continuous $f_i : [0,1]^{\xi_n} \to [0, 1]$ such that for all $\mcA \in \mbW_n$,
$\mcA(\psi(\bar{a}, b_i)) = f_i(\mfs_n(\mcA))$.
It follows that for all $\mcA \in \mbW_n$,
$\mcA(\varphi(\bar{a})) = F(\mcA(\psi(\bar{a}, b_1)), \dots, \mcA(\psi(\bar{a}, b_m))) =
F(f_1(\mfs_n(\mcA)), \ldots, f_m(\mfs_n(\mcA)))$.
Since $F$ is continuous as an aggregation function according to 
Definition~\ref{definition of continuous aggregation function}
it follows from part~(1) of that definition that $F$ is is continuous (in the usual sense of analysis in several variables)
when restricted to sequences from $[0, 1]^m$.
Hence the function $f : [0, 1]^{\xi_n} \to [0, 1]$ defined by 
$f(\bar{r}) = F(f_1(\bar{r}), \ldots, f_m(\bar{r}))$ for all $\bar{r} \in [0, 1]^{\xi_n}$ is continuous
and for all $\mcA \in \mbW_n$, $f(\mfs_n(\mcA)) = \mcA(\varphi(\bar{a}))$.
This concludes the proof.
\end{proof}

\begin{rem}\label{remark about measurability of asymptotic equivalence}{\rm
Let $\msfC : [0, 1]^2 \to [0, 1]$ be defined by $\msfC(r_1, r_2) = |r_1 - r_2|$, so $\msfC$ is continuous.
Let $\bar{x}$ be a sequence of distinct variables and let $\varphi_1(\bar{x}), \varphi_2(\bar{x}) \in CLA$.
Furthermore, let $\varphi(\bar{x})$ denote the formula $\msfC(\varphi_1(\bar{x}), \varphi_2(\bar{x}))$ of $CLA$.
By Lemma~\ref{definable sets are measurable}, 
it follows that, for all $n \in \mbbN^+$, all $\bar{a} \in [n]^{|\bar{x}|}$, and all $\varepsilon > 0$,
the set 
\[
\big\{\mcA \in \mbW_n : \big|\mcA(\varphi_1(\bar{a})) - \mcA(\varphi_2(\bar{a}))\big| \leq \varepsilon \big\}
\]
is a measurable subset of $\mbW_n$. 
As the set $[n]^{|\bar{x}|}$ is finite it follows that for every identity constraint $ic(\bar{x})$
the set 
\[
\bigcap_{\substack{\bar{a} \in [n]^{|\bar{x}|} \\ ic(\bar{a}) \ holds}}
\big\{\mcA \in \mbW_n : \big|\mcA(\varphi_1(\bar{a})) - \mcA(\varphi_2(\bar{a}))\big| \leq \varepsilon \big\}
\]
is a measurable subset of $\mbW_n$.
In particular it makes sense to talk about the probability of such a set,
using the probability distribution $\mbbP_n$ on $\mbW_n$.
}\end{rem}

\begin{defin}\label{definition of asymptotic equivalence}{\rm
Let $\bar{x}$ be a sequence of length $k$ of distinct variables and let $ic(\bar{x})$ be an identity constraint for $\bar{x}$.
The formulas $\varphi(\bar{x}), \psi(\bar{x}) \in CLA$ are {\bf \em asymptotically equivalent with respect to $ic$} 
if for all $\varepsilon > 0$,
\begin{align*}
\lim_{n\to\infty} \mbbP_n\Big( &\big\{\mcA \in \mbW_n : \text{ for all $\bar{a} \in [n]^k$ satisfying $ic(\bar{x})$,} \\
& \ \big|\mcA(\varphi(\bar{a})) - \mcA(\psi(\bar{a}))\big| \leq \varepsilon \big\} \Big) = 1.
\end{align*}
If both $\varphi$ and $\psi$ have no free variables, that is, if $\bar{x}$ is empty, then we omit the references to an
identity constraint.
}\end{defin}

\section{Main Results}

\noindent
The main results are given by the following theorem, which is subsequently proved in the rest of the section:

\begin{theor}\label{main result}
Let $\bar{x} = (x_1, \ldots, x_k)$ be a sequence of distinct variables, 
let $ic(\bar{x})$ be an identity constraint for $\bar{x}$, and let $\varphi(\bar{x}) \in CLA$.
(We allow the possiblility that $\bar{x}$ is empty in which case we can omit references to the identity constraint.)\\
(a) Then there is an aggregation-free formula $\psi(\bar{x}) \in CLA$ such that $\varphi(\bar{x})$ 
and $\psi(\bar{x})$ are asymptotically equivalent with respect to $ic(\bar{x})$.\\
(b) For every interval $I \subseteq [0, 1]$ there is $\alpha \in [0, 1]$ such that, for all $\varepsilon > 0$, if
$n$ is large enough and $\bar{a} \in [n]^{|\bar{x}|}$ satisfies $ic(\bar{x})$, then 
\begin{align*}
\big| \mbbP_n\big( \big\{ \mcA \in \mbW_n : \mcA(\varphi(\bar{a})) \in I \big\} \big) - \alpha \big| \leq \varepsilon.
\end{align*}
\end{theor}

\begin{rem}\label{remark on the main result}{\rm
It has been assumed that formulas in $CLA$ take values in $[0, 1]$ which is natural if we think of the values as probabilities
or degrees of truth or belief.
But for any constant $C > 1$ we can as well replace $[0, 1]$ by $[0, C]$ in all definitions, lemmas and proofs,
and the proofs will still work out.
}\end{rem}

\noindent
The rest of this section is devoted to proving Theorem~\ref{main result}.
We prove it by first proving a sequence of lemmas, then proving part~(a) of Theorem~\ref{main result}, and
then using part~(a) to prove part~(b) of the same theorem.
There are two key lemmas:
Lemma~\ref{probability of an aggregation-free formula and J} shows that 
the probability that an aggregation-free formula takes a value in an interval $J$ can be computed exactly by
computing an integral that depends only on the formula, an identity constraint and $J$.
Lemma~\ref{elimination of one aggregation function} 
shows how to asymptotically eliminate one aggregation function, which is later used in a proof that uses induction on
complexity of formulas.

\begin{lem}\label{primary form of the conditional probability}
Let $n, s \in \mbbN^+$, let $l_{1}, \ldots, l_{s} \in [\xi_n]$ be distinct, 
let $\msfC : [0, 1]^s \to [0, 1]$ be continuous, and let $J \subseteq [0, 1]$ be an interval.
Then 
\begin{align*}
&\mbbP_n^*\Big(\big\{ \bar{r} \in [0, 1]^{\xi_n} : \msfC(r_{l_1}, \ldots, r_{l_{s}}) \in J\big\} \Big) = \\
&\int_Y \lambda_{l_{1}}(r_1) \cdot \ldots \cdot \lambda_{l_{s}}(r_s) d r_1 \ldots d r_s
\quad \text{ where $Y := \msfC^{-1}(J)$.}
\end{align*}
\end{lem}

\noindent
{\bf Proof.}
Fix  $n, s \in \mbbN^+$ and let $l_1, \ldots, l_{s} \in [\xi_n]$, $\msfC$, and $J$ be as assumed.
For $i \in [\xi_n]$ let $I_i = [0, 1]$.
Let $Y = \msfC^{-1}(J)$.
As $\msfC : [0, 1]^s \to [0, 1]$ is continuous it follows that $Y$ is a (Lebesgue) measurable subset of $[0, 1]^s$.
As $I_i = [0, 1]$ for all $i \in [\xi_n]$
we may also view $\msfC$ as a function from $I_{l_{1}} \times \ldots \times I_{l_{s}}$ to $[0, 1]$ and
$Y$ may be identified with a measurable subset of  $I_{l_{1}} \times \ldots \times I_{l_{s}}$.
Now we get (where $\bar{r} = (r_1, \ldots, r_{\xi_n})$)
\begin{align*}
&\mbbP_n^*\Big(\big\{ \bar{r} \in [0, 1]^{\xi_n} : \msfC(r_{l_1}, \ldots, r_{l_{s}}) \in J \big\} \Big) = \\
&\mbbP_n^*\Big(\big\{ \bar{r} \in [0, 1]^{\xi_n} : (r_{l_1}, \ldots, r_{l_s}) \in Y \big\} \Big) =\\
&\int_{Y \times \prod_{i \in [\xi_n] \setminus \{l_1, \ldots, l_{s}\}} I_i} \ \pi_n^*(r_1, \ldots, r_{\xi_n}) d r_1 \ldots d r_{\xi_n} = \\
&\int_{Y \times \prod_{i \in [\xi_n] \setminus \{l_1, \ldots, l_{s}\}} I_i} \ 
\lambda_1(r_1) \cdot \ldots \cdot  \lambda_{\xi_n}(r_{\xi_n})  d r_1 \ldots d r_{\xi_n} = \\
&\int_Y \lambda_{l_{1}}(r_{l_{1}}) \cdot \ldots \cdot  \lambda_{l_{s}}(r_{l_{s}}) d r_{l_{1}} \ldots d r_{l_{s}}
\cdot \prod_{i \in [\xi_n] \setminus \{l_1, \ldots, l_{s}\}} \int_{I_i} \lambda_i(r_i) d r_i \ =  \\
&\int_Y \lambda_{l_{1}}(r_{l_{1}}) \cdot \ldots \cdot  \lambda_{l_{s}}(r_{l_{s}}) d r_{l_{1}} \ldots d r_{l_{s}}
\end{align*}
because $I_i = [0, 1]$ and $\int_{[0, 1]} \lambda_i(r_i) d r_i = 1$ for all $i$ 
(by Lemma~\ref{lambda-i and pi* are distributions}).
\hfill $\square$

\begin{lem}\label{probability of an aggregation-free formula and J}
Let $\bar{x} = (x_1, \ldots, x_k)$ be a sequence of distinct variables, let $ic(\bar{x})$ be an identity constraint for $\bar{x}$,
let $\msfC : [0, 1]^s \to [0, 1]$ be continuous, 
and let $\varphi(\bar{x})$ denote the formula
$\msfC\big(R_{1}(\bar{x}_{1}), \ldots, R_{s}(\bar{x}_{s})\big)$,
where, for all $i = 1, \ldots, s$, 
$R_i \in \sigma$, $\rng(\bar{x}_i) \subseteq \rng(\bar{x})$, $|\bar{x}_i| = \nu_{R_i}$,
and if $1 \leq i < j \leq s$
then $R_i \neq R_j$ or $\bar{x}_i \neq \bar{x}_j$. (We allow that $\bar{x}_i$ contains repetitions of variables.)
Let $J \subseteq [0, 1]$ be an interval,
let $n \in \mbbN^+$, let $\bar{a} = (a_1, \ldots, a_k) \in [n]^k$ satisfy $ic(\bar{x})$, and let,
for all $i = 1, \ldots, s$, $\bar{a}_i$ be the subsequence of $\bar{a}$ that corresponds to $\bar{x}_i$
(so if $\bar{x}_i = (x_{l_1}, \ldots, x_{l_{\nu_{R_i}}})$ then $\bar{a}_i = (a_{l_1}, \ldots, a_{l_{\nu_{R_i}}})$), and let $ic_i(\bar{x}_i)$
be the restriction of $ic(\bar{x})$ to $\bar{x}_i$.
Then
\begin{align*}
\mbbP_n\Big( \big\{ \mcA \in \mbW_n : \ \mcA(\varphi(\bar{a})) \in J \big\} \Big) = 
\int_Y \mu_{R_{1}}^{ic_1}(r_1) \cdot \ldots \cdot \mu_{R_{s}}^{ic_s}(r_s) d r_1 \ldots d r_s
\end{align*}
where $Y := \msfC^{-1}(J)$.
Hence the integral depends only on the formula $\varphi(\bar{x})$, the indentity constraint $ic(\bar{x})$ and $J$;
in particular there is no dependence on $n$.
\end{lem}

\noindent
{\bf Proof.}
We adopt the assumptions of the lemma, so in particular some $n$ is fixed and $\bar{a} \in [n]^k$ is a sequence 
that satisfies $ic(\bar{x})$.
By assumption, if $1 \leq i < j \leq s$
then $R_i \neq R_j$ or $\bar{x}_i \neq \bar{x}_j$.
Let $l_1, \ldots, l_{s}$ be the places of
the distinct pairs $(R_1, \bar{a}_1), \ldots, (R_{s}, \bar{a}_{s}) \in \Xi_n$ according to the order $\prec_n$,
so $l_1, \ldots, l_{s}$ are distinct numbers in $[\xi_n]$.

We now get (where $\bar{r} = (r_1, \ldots, r_{\xi_n})$)
\begin{align*}
&\mbbP_n\Big( \big\{ \mcA \in \mbW_n : \ \mcA(\varphi(\bar{a})) \in J \big\} \Big)  = 
\mbbP_n^*\Big(\big\{ \bar{r} \in [0, 1]^{\xi_n} : \msfC(r_{l_1}, \ldots, r_{l_{s}}) \in J\big\} \Big) \\
&= \ \int_Y \lambda_{l_{1}}(r_1) \cdot \ldots \cdot \lambda_{l_{s}}(r_s) d r_1 \ldots d r_s
\quad \text{(by Lemma~\ref{primary form of the conditional probability} where $Y := \msfC^{-1}(J)$)} \\
&= \ \int_Y \mu_{R_{1}}^{ic_1}(r_1) \cdot \ldots \cdot \mu_{R_{s}}^{ic_s}(r_s) d r_1 \ldots d r_s \\
&\text{  (by Definition~\ref{definition of mu-R etc} of $\lambda_i$ and the choice of $l_1, \ldots, l_{s}$).}
\end{align*}
\hfill $\square$

\begin{lem}\label{about independence}
Let $\bar{x} = (x_1, \ldots, x_k)$ be a sequence of distinct variables, let 
$y$ be a variable that does not occur in $\bar{x}$ and let $ic(\bar{x})$ be an identity constraint for $\bar{x}$.
Let $\msfC : [0, 1]^s \to [0, 1]$ be continuous 
and let $\varphi(\bar{x}, y)$ denote the formula
$\msfC\big(R_1(\bar{x}_1, y), \ldots, R_s(\bar{x}_s, y)\big)$,
where, for all $i = 1, \ldots, s$, 
$R_i \in \sigma$, $\rng(\bar{x}_i) \subseteq \rng(\bar{x})$, $|\bar{x}_i| = \nu_{R_i} - 1$,
and if $1 \leq i < j \leq s$
then $R_i \neq R_j$ or $\bar{x}_i \neq \bar{x}_j$.

Let $J \subseteq [0, 1]$ be an interval, let $n \in \mbbN^+$, let $\bar{a} \in [n]^{|\bar{x}|}$ satisfy $ic(\bar{x})$, 
and suppose that $b_1, \ldots, b_m \in [n]$ are distinct and do not occur in $\bar{a}$.
For all $i = 1, \ldots, m$, the event $\{\mcA \in \mbW_n : \mcA(\varphi(\bar{a}, b_i)) \in J \}$
is independent from all events $\{\mcA \in \mbW_n : \mcA(\varphi(\bar{a}, b_j)) \in J \}$ where $j \neq i$.
\end{lem}

\noindent
{\bf Proof.}
Fix $n \in \mbbN^+$ and a sequence $\bar{a} = (a_1, \ldots, a_k) \in [n]^k$ that satisfies $ic(\bar{x})$.
Let $\bar{a}_i$ be the subsequence of $\bar{a}$ that corresponds to the subsequence $\bar{x}_i$ of $\bar{x}$.
For each $i = 1, \ldots, m$, let $l_{i, 1}, \ldots, l_{i, s}$ be the places of
the (by assumption) distinct pairs $(R_1, \bar{a}_1b_i), \ldots, (R_{s}, \bar{a}_{s}b_i) \in \Xi_n$ according to the order $\prec_n$,
so $l_{i, 1}, \ldots, l_{i, s}$ are distinct numbers in $[\xi_n]$.
Since $b_i \neq b_j$ if $i \neq j$ it follows that $l_{i, l} \neq l_{j, l'}$ if $i \neq j$, $1 \leq l \leq s$ and $1 \leq l' \leq s$.

Note that for all $l = 1, \ldots, s$ and $i, j = 1, \ldots, m$,
$\bar{a}_l b_i$ and $\bar{a}_l b_j$ satisfy the same identity constraint.
For each $l = 1, \ldots, s$, let $id_l(\bar{x}_l, y)$ be the (up to equivalence unique) identity constraint that is satified
by $\bar{a}_l b_1$.

For each $i \in [\xi_n]$ let $I_i = [0, 1]$.
Let $Y = \msfC^{-1}(J)$, so $Y \subseteq [0, 1]^s$. 
For all $i = 1, \ldots, m$ let $Y_i \subseteq I_{l_{i, 1}} \times \ldots \times I_{l_{i, s}}$ be a copy of $Y$.
(That is $Y_i = \msfC^{-1}(J)$ if $\msfC$ is viewed as a function from $I_{l_{i, 1}} \times \ldots \times I_{l_{i, s}}$ to $[0, 1]$.)
Let $L = \bigcup_{i=1}^m \{l_{i, 1}, \ldots, l_{i, s}\}$.
For $i = 1, \ldots, m$, let $\mbE_n^i = \{\mcA \in \mbW_n : \mcA(\varphi(\bar{a}, b_i)) \in J \}$.
Now we have (where $\bar{r} = (r_1, \ldots, r_{\xi_n})$)
\begin{align*}
&\mbbP_n\Big(\bigcap_{i=1}^m \mbE_n^i\Big) = 
\mbbP_n^*\Big(\big\{ \bar{r} \in [0, 1]^{\xi_n} : \text{ for all $i = 1, \ldots, m$, } (r_{l_{i, 1}}, \ldots, r_{l_{i, s}}) \in J \big\} \Big) = \\
&\int_{Y_1 \times \ldots \times Y_m \times \prod_{i \in [\xi_n] \setminus L} I_i} \ \ 
\pi_n^*(r_1, \ldots, r_{\xi_n}) d r_1 \ldots d r_{\xi_n} = \\
&\int_{Y_1 \times \ldots \times Y_m \times \prod_{i \in [\xi_n] \setminus L} I_i} \ \ 
\lambda_1(r_1) \cdot \ldots \cdot \lambda_{\xi_n}(r_{\xi_n}) d r_1 \ldots d r_{\xi_n} = \\
&\bigg(\prod_{i=1}^m \int_{Y_i} \lambda_{l_{i, 1}}(r_{l_{i, 1}}) \cdot \ldots \cdot \lambda_{l_{i, s}}(r_{l_{i, s}})
d r_{l_{i, 1}} \ldots d r_{l_{i, s}}\bigg) \ \cdot \ 
\bigg(\prod_{i \in [\xi_n] \setminus L} \int_{I_i} \lambda_i(r_i) d r_i \bigg) = \\
&\prod_{i=1}^m \int_{Y_i} \lambda_{l_{i, 1}}(r_{l_{i, 1}}) \cdot \ldots \cdot \lambda_{l_{i, s}}(r_{l_{i, s}})
d r_{l_{i, 1}} \ldots d r_{l_{i, s}} = \\
&\prod_{i=1}^m \int_{Y_i} \mu_{R_1}^{id_1}(r_{l_{i, 1}}) \cdot \ldots \cdot \mu_{R_s}^{id_s}(r_{l_{i, s}})
d r_{l_{i, 1}} \ldots d r_{l_{i, s}} = \ \prod_{i=1}^m \mbbP_n\big(\mbE_n^i\big) 
\quad \text{ (by Lemma~\ref{probability of an aggregation-free formula and J}).}
\end{align*}
Hence the events $\mbE_n^1, \ldots, \mbE_n^m$ are independent.
\hfill $\square$

\medskip

\noindent
The following is a direct consequence of \cite[Corollary~A.1.14]{AS} which in turn follows from a bound given by
Chernoff \cite{Che}:

\begin{lem}\label{independent bernoulli trials}
Let $Z$ be the sum of $n$ independent 0/1-valued random variables, each one with probability $p$ of having the value 1,
where $p > 0$.
For every $\varepsilon > 0$ there is $c_\varepsilon > 0$, depending only on $\varepsilon$, such that the probability that
$|Z - pn| > \varepsilon p n$ is less than $2 e^{-c_\varepsilon p n}$.
(If $p = 0$ then the same statement holds if `$2 e^{-c_\varepsilon p n}$' is replaced by `$e^{-n}$'.)
\end{lem}

\begin{lem}\label{consequence of chernoff and independence}
Let $\bar{x} = (x_1, \ldots, x_k)$ be a sequence of distinct variables, let 
$y$ be a variable that does not occur in $\bar{x}$, let $ic(\bar{x})$ be an identity constraint for $\bar{x}$,
and let $ic'(\bar{x}, y)$ be the (up to equivalence unique) identity constraint for $(x_1, \ldots, x_k, y)$ which implies $ic(\bar{x})$
and $y \neq x_i$ for all $i = 1, \ldots, k$.
Let $\msfC : [0, 1]^s \to [0, 1]$ be continuous 
and let $\varphi(\bar{x}, y)$ denote the formula
$\msfC\big(R_1(\bar{x}_1, y), \ldots, R_s(\bar{x}_s, y)\big)$,
where, for all $i = 1, \ldots, s$, 
$R_i \in \sigma$, $\rng(\bar{x}_i) \subseteq \rng(\bar{x})$, $|\bar{x}_i| = \nu_{R_i} - 1$,
and if $1 \leq i < j \leq s$
then $R_i \neq R_j$ or $\bar{x}_i \neq \bar{x}_j$.

Let $J_1, \ldots, J_t \in [0, 1]$ be intervals.
For $i = 1, \ldots, t$, let $Y_i = \msfC^{-1}(J_i)$ and let
\[
\alpha_i = \int_{Y_i} \mu_{R_{1}}^{ic_1}(r_1) \cdot \ldots \cdot \mu_{R_{s}}^{ic_s}(r_s) d r_1 \ldots d r_s
\]
where, for $j = 1, \ldots, s$, $ic_j$ is the restriction of $ic'(\bar{x}, y)$ to $\bar{x}_jy$.
Let $\varepsilon > 0$.
Then there is $\beta > 0$ such that for all sufficiently large $n$ the probability that the following holds
for a random $\mcA \in \mbW_n$ is at least
$1 - e^{-\beta n}$:
For all $\bar{a} \in [n]^k$ that satisfy $ic(\bar{x})$ and all $i = 1, \ldots, t$,
\begin{equation}\label{proportion for alpha-i}
\alpha_i - \varepsilon  \leq 
\frac{\big| \big\{ b \in [n] \setminus \rng(\bar{a}) : \mcA(\varphi(\bar{a}, b)) \in J_i) \big\} \big|}{n - |\rng(\bar{a})|} 
\leq \alpha_i + \varepsilon.
\end{equation}
\end{lem}

\noindent
{\bf Proof.}
Fix some $n \in \mbbN^+$ and a sequence $\bar{a} \in [n]^k$ satisfying $ic(\bar{x})$.
Also fix some $i \in \{1, \ldots, t\}$.
Let $b_1, \ldots, b_m$ be a list, without repetition, of all members of $[n] \setminus \rng(\bar{a})$, so $m = n-|\rng(\bar{a})|$.
For all $j = 1, \ldots, m$, let $\mbE_n^j = \{\mcA \in \mbW_n : \mcA(\varphi(\bar{a}, b_j)) \in J_i \}$ and let
$Z_j : \mbW_n \to \{0, 1\}$ be the random variable defined by $Z_j(\mcA) = 1$ if $\mcA \in \mbE_n^j$ and $Z_i(\mcA) = 0$
otherwise. Let $Z = Z_1 + \ldots + Z_m$.
By Lemma~\ref{probability of an aggregation-free formula and J},
the probability that $Z_j = 1$ is $\alpha_i$.
By Lemma~\ref{about independence},
$Z_1, \ldots, Z_m$ are independent.
As $m = n-|\rng(\bar{a})|$ it  follows from Lemma~\ref{independent bernoulli trials} that for every $\varepsilon > 0$ there is $c > 0$
(depending only on $\varepsilon$)
such that the probability that $|Z - \alpha_i (n-|\rng(\bar{a})|)| > \varepsilon \alpha_i (n-|\rng(\bar{a})|)$
 is less than $2 e^{-c \alpha_i (n-|\rng(\bar{a})|)}$.
By assuming that $n$ is large enough and adjusting $c$ to some (slightly smaller) $c' > 0$ it follows that
the probability that $|Z - \alpha_i (n-|\rng(\bar{a})|)| > \varepsilon \alpha_i (n-|\rng(\bar{a})|)$ is less than $2 e^{-c' \alpha_i n}$.
This is equivalent to saying that the probability that a random $\mcA \in \mbW_n$ does 
{\em not} satisfy~(\ref{proportion for alpha-i})
is less than $2e^{-c' \alpha_i n}$. 
Since there are $n^k$ sequences $\bar{a} \in [n]^k$ it follows that the probability
that there is some $\bar{a} \in [n]^k$ satisfying $ic(\bar{x})$ such that~(\ref{proportion for alpha-i})
is not satisfied is less than $n^k \cdot  2 e^{-c' \alpha_i n}$.
Hence the probability that there is some $i \in \{1, \ldots, t\}$ and $\bar{a}$ such that~(\ref{proportion for alpha-i})
is not satisfied is less than $n^k \cdot  2 t e^{-c' \alpha_i n}$.
Then there is $\beta > 0$ such that $n^k \cdot 2 e^{-c' \alpha_i n} \leq e^{-\beta n}$ for all sufficiently large $n$.
Thus, the probability that, for sufficiently large $n$, 
a random $\mcA \in \mbW_n$ satisfies~(\ref{proportion for alpha-i}) for all $i = 1, \ldots, t$
and all sequences $\bar{a} \in [n]^k$ that satisfy $ic(\bar{x})$ is at least $1 - e^{-\beta n}$.
\hfill $\square$

\begin{lem}\label{elimination of one aggregation function}
Let $\bar{x} = (x_1, \ldots, x_k)$ be a sequence of distinct variables,
let $ic(\bar{x})$ be an identity constraint for $\bar{x}$, suppose that the variable $y$ does not occur in $\bar{x}$.
Suppose that $\varphi(\bar{x}, y) \in CLA$ is aggregation-free and that $F : [0, 1]^{<\omega} \to [0, 1]$ 
is a continuous aggregation function.
Then the formula $F(\varphi(\bar{x}, y) : y)$ is asymptotically equivalent to an aggregation-free formula with respect to $ic(\bar{x})$.
\end{lem}

\noindent
{\bf Proof.}
By
Lemma~\ref{reduction to simpler formula},
we may assume that $\varphi(\bar{x}, y)$ is a constant or has the form
\begin{equation*}
\msfC(R_1(\bar{x}_1), \ldots, R_t(\bar{x}_t), R_{t+1}(\bar{x}_{t+1}, y), \ldots, R_{t+s}(\bar{x}_{t+s}, y))
\end{equation*}
where $\msfC : [0, 1]^{t+s} \to [0, 1]$ is continuous, for all $i$,
$R_i \in \sigma$, $\rng(\bar{x}_i) \subseteq \rng(\bar{x})$, $|\bar{x}_i| = \nu_{R_i}$ if $i \leq t$
and $|\bar{x}_i| = \nu_{R_i} - 1$ otherwise, and
if $1 \leq i < j \leq t+s$
then $R_i \neq R_j$ or $\bar{x}_i \neq \bar{x}_j$.
The first case, when $\varphi(\bar{x}, y)$ is a constant, is covered by the second case since
$\msfC$ may be a constant function. Therefore we only consider the second case.

Let $\psi(\bar{x}) := F(\varphi(\bar{x}, y) : y)$.
It suffices to find continuous $\msfD : [0, 1]^t \to [0, 1]$
such that if $\theta(\bar{x}) := \msfD(R_1(\bar{x}_1), \ldots, R_t(\bar{x}_t))$ then,
for all $\varepsilon > 0$,
\begin{align}\label{desired property of D}
\lim_{n\to\infty} \mbbP_n\Big( &\big\{\mcA \in \mbW_n : \text{ for all $\bar{a} \in [n]^{|\bar{x}|}$ satisfying $ic(\bar{x})$} \\
& \ |\mcA(\psi(\bar{a})) - \mcA(\theta(\bar{a}))| \leq \varepsilon \big\} \Big) = 1. \nonumber
\end{align}
Fix $\bar{r} = (r_1, \ldots, r_t) \in [0, 1]^t$. 
We now find out what $\msfD(\bar{r})$ should be for~(\ref{desired property of D}) to hold.
First we define $\msfC_{\bar{r}} : [0, 1]^s \to [0, 1]$ by 
\[
\msfC_{\bar{r}}(p_1, \ldots, p_s) = \msfC(r_1, \ldots, r_t, p_1, \ldots, p_s)
\]
so $\msfC_{\bar{r}}$ is continuous.
Let $\varphi_{\bar{r}}(\bar{x}, y) := \msfC_{\bar{r}}(R_{t+1}(\bar{x}_{t+1}, y), \ldots, R_{t+s}(\bar{x}_{t+s}, y))$.
Let $ic'(\bar{x}, y)$ be the (unique up to equivalence) identity constraint that implies $ic(\bar{x})$ and $y \neq x_i$ for all 
$i = 1, \ldots, k$.

Let $M \in \mbbN^+$ and $J_i := [\frac{i}{M}, \frac{i+1}{M}]$ for $i = 0, \ldots, M-1$.
By Lemma~\ref{probability of an aggregation-free formula and J},
for all $i = 0, \ldots, M-1$, there is $\alpha_i$ depending only on $ic'(\bar{x}, y)$
(which depends only on $ic(\bar{x})$), $J_i$ and $\varphi_{\bar{r}}$
(which depends only on $\varphi$ and $\bar{r}$) such that for all $n$,
all tuples $\bar{a} \in [n]^k$ satisfying $ic(\bar{x})$ and all $b \in [n] \setminus \rng(\bar{a})$,
\[
\mbbP_n\big(\big\{\mcA \in \mbW_n : \mcA(\varphi_{\bar{r}}(\bar{a}, b)) \in J_i \big\}\big) = \alpha_i.
\]
For all $n$ and $\delta > 0$
let $\mbX_n^\delta$ be the set of all 
$\mcA \in \mbW_n$ such that
for all tuples $\bar{a} \in [n]^k$ that satisfy $ic(\bar{x})$ and all $i = 0, \ldots, M-1$,
\begin{equation}\label{the proportion of elements in J}
\alpha_i - \delta  \leq 
\frac{\big| \big\{ b \in [n] \setminus \rng(\bar{a}) : \mcA(\varphi_{\bar{r}}(\bar{a}, b)) \in J_i) \big\} \big|}{n - |\rng(\bar{a})|} 
\leq \alpha_i + \delta.
\end{equation}
According to Lemma~\ref{consequence of chernoff and independence},
$\lim_{n\to\infty}\mbbP_n\big(\mbX_n^\delta\big) = 1$.

Let $\varepsilon > 0$.
By assumption, $F$ is a continuous aggregation function.
This means that
$\delta> 0$ and $M, N \in \mbbN$ can be chosen so that if 
$\bar{q} = (	q_1, \ldots, q_n), \bar{q}' = (q'_1, \ldots, q'_m) \in [0, 1]^{<\omega}$
satisfy conditions (2)(a)--(d) of 
Definition~\ref{definition of continuous aggregation function} of continuity of aggregation functions,
then
$|F(\bar{q}) - F(\bar{q}')| \leq \varepsilon$.
Suppose that $\delta$, $M$ and $N$ are such numbers.
Without loss of generality we may assume that, for all $i = 1, \ldots, M-1$, if $\alpha_i > 0$ then $\alpha_i > \delta$.
Without loss of generality we may assume that $N \in \mbbN$ is large enough so that 
if $n \geq N + |\rng(\bar{a})|$ then $\mbX_n^\delta \neq \es$.
Let $n \geq N + |\rng(\bar{a})|$,
$\mcA \in \mbX_n^\delta$, 
and let $\bar{a} \in [n]^k$ satisfy $ic(\bar{x})$.
Enumerate $[n] \setminus \rng(\bar{a})$ as $b_1, \ldots, b_m$ where $m := n-|\rng(\bar{a})|$ so $m \geq N$.
Define $q_i := \mcA(\varphi_{\bar{r}}(\bar{a}, b_i))$ and $\bar{q} := (q_1, \ldots, q_m)$.
Then conditions (2)(a) and (2)(b) of
Definition~\ref{definition of continuous aggregation function}
hold for $\bar{q}$ (because $m \geq N$, $\alpha_i > \delta$ if $\alpha > 0$, and $\mcA \in \mbX_n^\delta$).
 
We now verify that conditions (2)(c) and (2)(d) of 
Definition~\ref{definition of continuous aggregation function}
hold for $\bar{q}$.
For all $j = 1, \ldots, s$ let $ic'_j(\bar{x}_j, y)$ be the restriction of $ic'(\bar{x}, y)$ to $\bar{x}_jy$.
By
Lemma~\ref{probability of an aggregation-free formula and J}
we have that, for all $i = 0, \ldots, M-1$,
\[
\alpha_i =  \int_{Y_i}  \mu_{R_{t+1}}^{ic'_1}(r_1) \cdot \ldots \cdot \mu_{R_{t+s}}^{ic'_s}(r_s) d r_1 \ldots d r_s
 \]
where 
$Y_i = \msfC_{\bar{r}}^{-1}(J_i) \subseteq [0, 1]^s$ and 
$g(r_1, \ldots, r_s) :=   \mu_{R_{t+1}}^{ic'_1}(r_1) \cdot \ldots \cdot \mu_{R_{t+s}}^{ic'_s}(r_s)$ 
is a continuous function, so $Y_i$ is a measurable subset of $[0, 1]^s$.
Note that by the continuity of $g$, $g([0, 1]^s)$ is an interval.
Hence $g([0, 1]^s) \cap J_i$ is an interval for all $i$.
Also, if $\alpha_i = 0$ then $g([0, 1]^s) \cap J_i$ contains at most one point; hence, if $g([0, 1]^s) \cap J_i$ contains more than
one point then $\alpha_i > 0$.
It follows that if $i < j < k$ (or $i > j > k$), $\alpha_i > 0$, $\alpha_j = 0$ and $\alpha_k > 0$ then
$g([0, 1]^s)$ is not an interval, so $g$ is {\it not} continuous.
Thus we conclude that condition (2)(c) of Definition~\ref{definition of continuous aggregation function} holds for $\bar{q}$.
Now suppose that $i < j < k$ (or $i > j > k$) $\alpha_i = \alpha_j = 0$ and $\alpha_k > 0$.
For a contradiction, suppose that $\{l : q_l \in J_i\} \neq \es$.
From $\alpha_k > 0$ it follows that $g([0, 1]^s) \cap J_k \neq \es$.
From $\alpha_j = 0$ it follows that $g([0, 1]^s) \cap J_i$ contains at most one point.
But then $g([0, 1]^s)$ is not an interval, contradicting that $g$ is continuous.
Hence, also condition (2)(d) of 
Definition~\ref{definition of continuous aggregation function} holds for $\bar{q}$.

By the assumptions on $\delta$, $M$ and $N$, 
it follows that there is an open interval $I_\varepsilon \subseteq [0, 1]$
with width (the distance between its endpoints) at most $\varepsilon$ and depending only on 
$\varphi, \bar{r}, M, N$, and $\delta$
such that $F(\bar{q}) \in I_\varepsilon$.
Furthermore, if $\varepsilon > \varepsilon' > 0$ it follows in the same way that 
there are $M' \geq M$, $N' \geq N$, and $0 \leq \delta' \leq \delta$ 
such that if $J'_i$, $\alpha'_i$ and $I_{\varepsilon'}$ are defined as $J_i$, $\alpha_i$, and $I_{\varepsilon}$, but with
$M'$, $N'$ and $\delta'$ in place of $M$, $N$ and $\delta$, respectively, then 
$I_{\varepsilon'} \subset I_\varepsilon$
and $I_{\varepsilon'}$ has width at most $\varepsilon'$.
We now define $D(\bar{r})$ to be the unique member of $\bigcap_{\varepsilon > 0} I_\varepsilon$.
It follows that if $\mcA \in \mbX_n^\delta$ and $R_i^{\mcA}(\bar{a}_i) = r_i$ for $i = 1, \ldots, t$,
then $\mcA(\psi(\bar{a})) = \mcA\big(F(\varphi(\bar{a}, y) : y)\big) =$ 
$\mcA\big(F(\varphi_{\bar{r}}(\bar{a}, y) : y)\big) = F(\bar{q})$
and $\mcA(\theta(\bar{a})) = \msfD(\bar{r})$; hence 
$|\mcA(\psi(\bar{a})) - \mcA(\theta(\bar{a}))| = |F(\bar{q}) - \msfD(\bar{r})| \leq \varepsilon$.
Observe that the definition of $\msfD(\bar{r})$ does not depend on the choice $\mcA$ or $\bar{a} \in [n]^k$
as long as $\mcA \in \mbX_n^\delta$ and $\bar{a}$ satisfies $ic(\bar{x})$, 
because~(\ref{the proportion of elements in J}) holds for all 
$\mcA \in \mbX_n^\delta$ and all $\bar{a} \in [n]^k$ that satisfy $ic(\bar{x})$.

For the continuity of $\msfD$, let $\bar{r}$ be as above and suppose that $\bar{r}' = (r'_1, \ldots, r'_t) \in [0, 1]^t$ and that
$|r_i - r'_i| \leq \delta$ for all $i = 1, \ldots, t$.
Since $\msfC$ is continuous on the compact space $[0, 1]^{t+s}$ it is in fact uniformly continuous.
It follows that, for every $\varepsilon' > 0$, if $\delta$ is small enough and,
for all $i = t+1, \ldots, t+s$, the reals
$r_i, r'_i \in [0, 1]$ are such that $|r_i - r'_i| \leq \delta$, then 
$|\msfC_ {\bar{r}}(r_{t+1}, \ldots, r_{r+s}) - \msfC_ {\bar{r}'}(r'_{t+1}, \ldots, r'_{r+s})| \leq \varepsilon'$.
Let $\varphi_{\bar{r}'}(\bar{x}, y) := \msfC_{\bar{r}'}(R_{t+1}(\bar{x}_{t+1}, y), \ldots, R_{t+s}(\bar{x}_{t+s}, y))$
For $\mcA$, $\bar{a}$ and $b_i$ as above we get
\[
\big|\mcA(\varphi_{\bar{r}}(\bar{a}, b_i)) - \mcA(\varphi_{\bar{r}'}(\bar{a}, b_i))\big| \leq \varepsilon'.
\]
For $i = 1, \ldots, m$ let $q'_i = \mcA(\varphi_{\bar{r}'}(\bar{a}, b_i))$ so we have $|q_i - q'_i| \leq \varepsilon'$ for all $i$.
Let $\bar{q}' = (q'_1, \ldots, q'_m)$.
Then, for every $\varepsilon > 0$, if $\varepsilon'$ is small enough (which can be arranged by taking $\delta$ small enough)
we get $|F(\bar{q}) - F(\bar{q}')| \leq \varepsilon$ because $F$ is continuous.
Given any $\varepsilon > 0$ we can argue as before (for $\bar{q}$ and $\bar{r}$)
to get $|F(\bar{q}') - \msfD(\bar{r}')| \leq \varepsilon$ (given that $\delta$
is small enough).
Then we have $|F(\bar{q}) - F(\bar{q}')| \leq \varepsilon$,  $|F(\bar{q}') - \msfD(\bar{r}')| \leq \varepsilon$, and
 $|F(\bar{q}) - \msfD(\bar{r})| \leq \varepsilon$, which gives $|\msfD(\bar{r}) - \msfD(\bar{r}')| \leq 3\varepsilon$.
Hence $\msfD$ is continuous.
\hfill $\square$

\begin{lem}\label{F of aggregation-free and not aggregation-free}
Let $\bar{x} = (x_1, \ldots, x_k)$ be a sequence of distinct variables and let $y \notin \rng(\bar{x})$ be a variable.
Let $ic(\bar{x}, y)$ be an identity constraint for $\bar{x}y$ which implies that $y \neq x_i$ for all $i = 1, \ldots, k$,
and let $ic'(\bar{x})$ be the restriction of $ic$ to $\bar{x}$.
Suppose that $\varphi(\bar{x}, y), \varphi'(\bar{x}, y) \in CLA$ 
where $\varphi(\bar{x},y)$ and $\varphi'(\bar{x}, y)$ are asymptotically equivalent with respect to $ic(\bar{x}, y)$ and 
$F : [0, 1]^{<\omega} \to [0, 1]$ is a continuous aggregation function.
Then $F(\varphi(\bar{x}, y) : y)$ and $F(\varphi'(\bar{x}, y) : y)$ are asymptotically equivalent with respect to $ic'(\bar{x})$.
\end{lem}

\noindent
{\bf Proof.}
For all $n \in \mbbN^+$ and $\delta > 0$ let
$\mbY_n^\delta$ be the set of all $\mcA \in \mbW_n$ such that for all $\bar{a} \in [n]^k$ that satisfy $ic'(\bar{x})$
and all $b \in [n] \setminus \rng(\bar{a})$,
\[
\big| \mcA(\varphi(\bar{a}, b)) - \mcA(\varphi'(\bar{a}, b)) \big| \leq \delta.
\]
Since we assume that $\varphi(\bar{x},y)$ and $\varphi'(\bar{x}, y)$ are asymptotically equivalent with respect to $ic(\bar{x}, y)$
it follows that for all $\delta > 0$, $\lim_{n\to\infty} \mbbP_n\big(\mbY_n^\delta\big) = 1$.
Let $\varepsilon > 0$.
It now suffices to show that if $\delta > 0$ is small enough, $n$ large enough, $\mcA \in \mbY_n^\delta$,
and $\bar{a} \in [n]^k$ satisfies $ic'(\bar{x})$, then
\begin{equation}\label{F(of) - F(of) is small}
\big| \mcA\big(F(\varphi(\bar{a}, y) : y)\big) - \mcA\big(F(\varphi'(\bar{a}, y) : y)\big) \big| \leq \varepsilon.
\end{equation}
So take $\mcA \in \mbY_n^\delta$ and $\bar{a} \in [n]^k$ that satisfies $ic'(\bar{x})$ and let
$b_1, \ldots, b_m$ enumerate $[n] \setminus \rng(\bar{a})$.
For $i = 1, \ldots, m$, let $q_i = \mcA(\varphi(\bar{a}, b_i))$ and $q'_i =  \mcA(\varphi'(\bar{a}, b_i))$, and then let
$\bar{q} = (q_1, \ldots, q_{m})$ and $\bar{q}' = (q'_1, \ldots, q'_{m})$.
Then $|q_i - q'_i| \leq \delta$ for all $i$.
Since we assume that $F$ is continuous it follows from 
condition~(1) of Definition~\ref{definition of continuous aggregation function} of continuity
that if $\delta$ is small enough and $n$ large enough then
$|F(\bar{q}) - F(\bar{q}')| \leq \varepsilon$, hence~(\ref{F(of) - F(of) is small}) holds.
\hfill $\square$

\subsection*{Proof of Theorem~\ref{main result}}
We begin with part (a) and prove it by induction on the complexity, or construction,  of $\varphi(\bar{x}) \in CLA$.
If $\varphi(\bar{x})$ has one of the forms 1, 2, or 3, in Definition~\ref{syntax of CLA}
then it is already aggregation-free and asymptotically equivalent to itself with respect to $id(\bar{x})$, 
for every identity constraint $ic(\bar{x})$.
Next, suppose that $\varphi(\bar{x})$ has the form 
$\msfC(\varphi_1(\bar{x}), \ldots, \varphi_k(\bar{x}))$ for some $\varphi_i(\bar{x}) \in CLA$, $i = 1, \ldots, k$,
and continuous $\msfC : [0, 1]^k \to [0, 1]$. 
Let $ic(\bar{x})$ be an identity constraint for $\bar{x}$.
By the induction hypothesis, for $i = 1, \ldots, k$, $\varphi_i(\bar{x})$
is asymptotically equivalent to some aggregation-free $\varphi'_i(\bar{x})$ with respect to $ic(\bar{x})$.
Then $\varphi'(\bar{x}) := \msfC(\varphi'_1(\bar{x}), \ldots, \varphi'_k(\bar{x}))$ is aggregation-free.
Since $\msfC$ is continuous it follows straightforwardly that $\varphi'(\bar{x})$ is asymptotically equivalent to $\varphi(\bar{x})$
with respect to $ic(\bar{x})$.

Now suppose that $\varphi(\bar{x})$ has the form $F(\psi(\bar{x}, y) : y)$ where $F$ is a continuous aggregation function
(and $\psi(\bar{x}, y) \in CLA$).
Let $ic(\bar{x})$ be an identity constraint on $\bar{x}$ and let $ic'(\bar{x}, y)$ be (the unique up to equivalence) 
identity constraint that implies $ic(\bar{x})$ and $y \neq x_i$ for all $i = 1, \ldots, k$.
By the induction hypothesis, $\psi(\bar{x}, y)$ is asymptotically equivalent to some aggregation-free $\psi'(\bar{x}, y)$
with respect to $ic'(\bar{x}, y)$.
By Lemma~\ref{elimination of one aggregation function},
$F(\psi'(\bar{x}, y) : y)$ is asymptotically equivalent to some aggregation-free $\theta(\bar{x})$ with respect to $ic(\bar{x})$.
Lemma~\ref{F of aggregation-free and not aggregation-free}
implies that $F(\psi(\bar{x}, y) : y)$ and $F(\psi'(\bar{x}, y) : y)$ are asymptotically equivalent with respect to $ic(\bar{x})$.
Since asymptotic equivalence (with respect to $ic(\bar{x})$) 
is a transitive relation between formulas (which is easily verified) it follows that
$F(\psi(\bar{x}, y) : y)$ is asymptotically equivalent to $\theta(\bar{x})$ with respect to $ic(\bar{x})$.

Finally we prove part (b) of Theorem~\ref{main result}.
Let $\varphi(\bar{x}) \in CLA$ and let $ic(\bar{x})$ be an identity constraint for $\bar{x}$.
By part~(a), $\varphi(\bar{x})$ is asymptotically equivalent to some aggregation-free $\psi(\bar{x})$.
By Lemma~\ref{reduction to simpler formula}
we may assume that $\psi(\bar{x})$ satisfies the preconditions of 
the formula in
Lemma~\ref{probability of an aggregation-free formula and J},
so that lemma is applicable to $\psi(\bar{x})$.

Let $I = [c, d] \subseteq [0, 1]$ be an interval.
The cases when $c = d$ or when $c=0$ and $d=1$ are trivial, so we assume that $0 < c < d < 1$.
(The cases when $c = 0$ or $d = 1$ are treated similarly, with minor modifications of the definitions of  
$I_\delta^+$ and $I_\delta^-$ below.)
Let $\delta > 0$ be such that $\delta < \min\big((d-c)/2, c/2, (1-d)/2\big)$ and let , 
$I_\delta^+ = [c - \delta, d + \delta]$ and $I_\delta^- = [c + \delta, d - \delta]$.
Lemma~\ref{probability of an aggregation-free formula and J} 
implies that there are $\alpha$, $\alpha_\delta^+$ and $\alpha_\delta^-$ such that for all $n$
and all $\bar{a} \in [n]^{|\bar{x}|}$ that satisfy $ic(\bar{x})$,
\begin{align}\label{the probabilities of psi}
&\mbbP_n\big( \big\{ \mcA \in \mbW_n : \mcA(\psi(\bar{a})) \in I \big\} \big) = \alpha, \\
&\mbbP_n\big( \big\{ \mcA \in \mbW_n : \mcA(\psi(\bar{a})) \in I_\delta^+ \big\} \big) = \alpha_\delta^+, \text{ and} 
\nonumber \\
&\mbbP_n\big( \big\{ \mcA \in \mbW_n : \mcA(\psi(\bar{a})) \in I_\delta^- \big\} \big) = \alpha_\delta^-,
\nonumber
\end{align}
where, for some $s \in \mbbN^+$, $R_1, \ldots, R_s \in \sigma$,
identity constraints $ic_1(\bar{x}_1), \ldots, ic_s(\bar{x}_s)$, and 
continuous $\msfC : [0, 1]^s \to [0, 1]$, 
\begin{align*}
&\alpha_\delta^+ = \int_{\msfC^{-1}(I_\delta^+)} 
\mu_{R_1}^{ic_1}(r_1) \cdot \ldots \cdot \mu_{R_s}^{ic_s}(r_s) d r_1 \ldots d r_s, \text{ and} \\
&\alpha_\delta^- = \int_{\msfC^{-1}(I_\delta^-)} 
\mu_{R_1}^{ic_1}(r_1) \cdot \ldots \cdot \mu_{R_s}^{ic_s}(r_s) d r_1 \ldots d r_s.
\end{align*}
As $I_\delta^- \subseteq I \subseteq I_\delta^+$ we get $\alpha_\delta^- \leq \alpha \leq \alpha_\delta^+$.
Let $\mbX_n^\delta$ be the set of all $\mcA \in \mbW_n$ such that
for all $\bar{a} \in [n]^{|\bar{x}|}$ that satisfy $ic(\bar{a})$, 
$|\mcA(\varphi(\bar{a})) - \mcA(\psi(\bar{a}))| \leq \delta$.
Since for all $n \in \mbbN^+$ and $\bar{a} \in [n]^{|\bar{x}|}$ that satisfy $ic(\bar{a})$,
\[
 \big\{ \mcA \in \mbX_n^\delta : \mcA(\psi(\bar{a})) \in I_\delta^- \big\}  \subseteq 
 \big\{ \mcA \in \mbX_n^\delta : \mcA(\varphi(\bar{a})) \in I \big\}  \subseteq
 \big\{ \mcA \in \mbX_n^\delta : \mcA(\psi(\bar{a})) \in I_\delta^+ \big\}
\]
we must have
\begin{align}\label{inequalities of probabilities}
&\mbbP_n\big( \big\{ \mcA \in \mbX_n^\delta : \mcA(\psi(\bar{a})) \in I_\delta^- \big\} \big)
\leq \\
&\mbbP_n\big( \big\{ \mcA \in \mbX_n^\delta : \mcA(\varphi(\bar{a})) \in I \big\} \big) 
\leq \nonumber \\
&\mbbP_n\big( \big\{ \mcA \in \mbX_n^\delta : \mcA(\psi(\bar{a})) \in I_\delta^+ \big\} \big). 
\nonumber
\end{align}
Since $\lim_{n\to\infty}\mbbP_n\big(\mbX_n^\delta\big) = 1$ it follows 
from~(\ref{the probabilities of psi})
that, for every $\varepsilon > 0$ and
all sufficiently large $n$
and all $\bar{a} \in [n]^{|\bar{x}|}$ that satisfy $ic(\bar{x})$, we have
\begin{align*}
&\mbbP_n\big( \big\{ \mcA \in \mbX_n^\delta : \mcA(\psi(\bar{a})) \in I_\delta^+ \big\} \big) = 
(\alpha_\delta^+ - \varepsilon, \alpha_\delta^+ + \varepsilon), \text{ and} \\
&\mbbP_n\big( \big\{ \mcA \in \mbX_n^\delta : \mcA(\psi(\bar{a})) \in I_\delta^- \big\} \big) = 
(\alpha_\delta^- - \varepsilon, \alpha_\delta^- + \varepsilon).
\end{align*}
This and~(\ref{inequalities of probabilities}) gives, for sufficiently large $n$,
\begin{equation*}
\alpha_\delta^-  - \varepsilon <
\mbbP_n\big( \big\{ \mcA \in \mbX_n^\delta : \mcA(\varphi(\bar{a})) \in I \big\} \big)  <
\alpha_\delta^+ + \varepsilon.
\end{equation*}
Since $\lim_{n\to\infty}\mbbP_n\big(\mbX_n^\delta\big) = 1$ it follows that for all sufficiently large $n$
and all $\bar{a} \in [n]^{|\bar{x}|}$ that satisfy $ic(\bar{x})$,
\begin{equation}\label{almost final inequalities}
\alpha_\delta^-  - 2\varepsilon <
\mbbP_n\big( \big\{ \mcA \in \mbW_n : \mcA(\varphi(\bar{a})) \in I \big\} \big)  <
\alpha_\delta^+ + 2\varepsilon.
\end{equation}
Recall that $\msfC^{-1}(I_\delta^-) \subseteq \msfC^{-1}(I_\delta^+)$ and 
$\alpha_\delta^- \leq \alpha \leq \alpha_\delta^+$.
By taking $\delta > 0$ small enough we can make the measure of
of $\msfC^{-1}(I_\delta^+) \setminus \msfC^{-1}(I_\delta^-)$ as small as we like.
Due to the expressions above of $\alpha_\delta^+$ and $\alpha_\delta^-$ as integrals over
$\msfC^{-1}(I_\delta^+)$ and $\msfC^{-1}(I_\delta^-)$, respectively, it follows that for all $\varepsilon > 0$,
if $\delta > 0$ is sufficiently small then $\alpha_\delta^+ - \alpha_\delta^- < \varepsilon$.
So from~(\ref{almost final inequalities}) we get 
\[
\big| \mbbP_n\big( \big\{ \mcA \in \mbW_n : \mcA(\varphi(\bar{a})) \in I \big\} \big) - \alpha \big| \ \leq \ 4\varepsilon
\]
if $\delta > 0$ is small enough and $n$ is large enough and this completes the proof.
\hfill $\square$

\subsection*{Final remarks}
One can ask if Theorem~\ref{main result} can be generalized in various ways.
For example, can a similar theorem be proved if the $CLA$ is allowed to use more general forms of aggregations
(as in \cite{Jae98a, KW1, KW2}), for
example over tuples of elements, or over elements or tuples that satisfy some condition? 
Another conceivable generalization would be to allow the probability distribution to be generated by some form
of probabilistic graphical model \cite{Koller, BKNP, KMG, RKNP}, as in e.g. \cite{Jae98a, KW1, KW2}, but adapted to real valued relations. 

\subsection*{Acknowledgements}
Vera Koponen thanks the anonymous referee for valuable comments including pointing out a couple of
gaps in the arguments of the initial manuscript which have now been resolved.
Vera Koponen was partially supported by the Swedish Research Council, grant 2023-05238\_VR.


\begin{thebibliography}{99}\label{References}

\bibitem{Adam-Day1} Sam Adam-Day, Theodor-Mihai Iliant and Ismail Ilkan Ceylan,
Zero-one laws of graph neural networks,
{\em Proceedings of 37th Conference on Neural Information Processing Systems (NeurIPS 2023)}.

\bibitem{Adam-Day2} Sam Adam-Day, Michael Benedikt, Ismail Ilkan Ceylan and Ben Finkelshtein,
Almost surely asymptotically constant graph neural networks,
{\em Proceedings of the 38th Conference on Neural Information Processing Systems (NeurIPS 2024)}.

\bibitem{ABFN} Michael Albert, Mathilde Bouvel, Valentin F\'{e}ray and Marc Noy,
A logical limit law for 231-avoiding permutations, 
{\em Discrete Mathematics and Theoretical Computer Science}, Vol. 26:1, Permutation Patterns 2023 \#1 (2024).

\bibitem{AS} Noga Alon and Joel H. Spencer,
{\em The Probabilistic Metod, Second Edition},
John Wiley \& Sons (2000).

\bibitem{Bal} J. T. Baldwin, Expansions of geometries, 
{\em The Journal of Symbolic Logic}, Vol. 68 (2003) 803--827.


\bibitem{BenY} I. Ben Yaacov, A. Berenstein, C. W. Henson and A. Usvyatsov,
Model theory for metric structures, in
Z. Chatzidakis, D. Machpherson, A. Pillay and A. Wilkie (editors),
{\em Model theory with applications to algebra and analysis,
London Mathematical Society Lecture Notes Series}, 
Vol. 350 (2008) 315--427.


\bibitem{BKNP}
Guy van den Broeck and Kristian Kerstin and Sriram Natarajan and David Poole (editors),
{\em An Introduction to Lifted Probabilistic Inference},
The MIT Press
 (2021).
 
 
\bibitem{Bur} S. N. Burris, {\em Number Theoretic Density and Logical Limit Laws},
Mathematical Surveys and Monographs, Vol. 86, American Mathematical Socitety (2001).

\bibitem{CK}
C. C. Chang and H. Jerome Keisler,
  {\em Continuous Model Theory},
Princeton University Press 
 (1966).

\bibitem{Che} H. Chernoff, A measure of the asymptotic efficiency for tests of a hypothesis based on
the sum of observations, {\em Annals of Mathematical Statistics}, Vol. 23 (1952) 493--509.


\bibitem{CM}
 Fabio G. Cozman and Denis D. Mau\'{a},
 The finite model theory of Bayesian network specifications:
  Descriptive complexity and zero/one laws,
 {\em International Journal of Approximate Reasoning},
  Vol. 110 (2019)
 107--126.
 
 
 \bibitem{RKNP}
 Luc De Raedt and Kristian Kersting and Sriraam Natarajan and David Poole,
  {\em Statistical Relational Artificial Intelligence; Logic, Probability, and Computation},
 Morgan and Claypool Publishers
  (2016).


\bibitem{EF}
 Heinz-Dieter Ebbinghaus and Jörg Flum,
  {\em Finite Model Theory, Second Edition},
  Springer
  (1999).

\bibitem{Fag}
  Ronald Fagin,
  Probabilities on finite models,
   {\em The Journal of Symbolic Logic},
  Vol. 41 (1976)
  50--58.



\bibitem{Gle}
 Y. Glebskii, D. Kogan, M Liogon'kii and V. Talanov,
  Range and degree of realizability of formulas in the restricted predicate calculus,
  {\em Cybernetics},
  Vol. 5 (1969)
 142--154.


\bibitem{Gol} Isaac Goldbring, Bradd Hart and Alex Kruckman,
The almost sure theory of finite metric spaces,
{\em Bulletin of the London Mathematical Society},
Vol. 53 (2021) 1740--1748.

\bibitem{Gra}
 Erich Grädel, Hayyan Helal, Matthias Naaf and Richard Wilke,
  Zero-One Laws and Almost Sure Valuations of First-Order Logic in Semiring Semantics,
  {\em Proceedings of the 37th Annual ACM/IEEE Symposium on Logic in Computer Science (LICS 22)},
  (2022) 1--12.

\bibitem{HK} S. Haber and M. Krivelevich, The logic of random regular graphs, {\em Journal of Combinatorics},
Vol. 1 (2010) 389--440.

\bibitem{Hill}  C. D. Hill, On 0, 1-laws and asymptotics of definable sets in geometric Fra\"{i}ss\'{e} classes, 
{\em Fundamenta Mathematica}, Vol. 239 (2017) 201--219.

\bibitem{Jae98a}
  Manfred Jaeger,
  Convergence results for relational Bayesian networks, 
  {\em Proceedings of the Thirteenth Annual IEEE Symposium on Logic in Computer Science (LICS 98)},
  (1998) 44--55.

\bibitem{Jae98b}
  Manfred Jaeger,
  Reasoning about infinite random structures with relational Bayesian networks,
  {\em Proceedings of the Sixth International Conference on Principles of Knowledge Representation and Reasoning},
  (1998) 570--581.


\bibitem{Kai}
  Risto Kaila,
  On probabilistic elimination of generalized quantifiers,
  {\em Random Structures and Algorithms},
  Vol. 19 (2001)
 1--36.


\bibitem{KL}
  H. Jerome Keisler and Wafik B. Lotfallah,
  Almost everywhere elimination of probability quantifiers,
  {\em The Journal of Symbolic Logic},
  Vol 74 (2009)
  1121--1142.

\bibitem{KMG}
 Angelika Kimmig, Lilyana Mihalkova and Lise Getoor,
  Lifted graphical models: a survey,
  {\em Machine Learning},
  Vol. 99 (2015) 1--45.

\bibitem{KoVa} Ph. G. Kolaitis and M. Y. Vardi, Infinitary logics and 0-1 laws,
{\em Information and Computation}, 
Vol 98 (1992) 258--294.

\bibitem{KPR} Ph. G. Kolaitis, H. J. Prömel and B. L. Rotchild, $K_{l+1}$-free graphs: asymptotic structure and a 0-1 law,
{\em Transactions of the American Mathematical Society}, Vol. 303 (1987) 637--671.

\bibitem{Koller} D. Koller, N. Friedman,
{\em Probabilistic Graphical Models: Principles and Techniques},
MIT Press (2009).


\bibitem{Kop20}
  Vera Koponen,
  Conditional probability logic, lifted Bayesian networks, and almost sure quantifier elimination,
  {\em Theoretical Computer Science},
  Vol. 848 (2020)
  1--27.

 \bibitem{Kop26}
  Vera Koponen,
  Random expansions of finite structures with bounded degree,
  {\em Annals of Pure and Applied Logic}, Vol. 177 (2026) 103665,
  \url{https://doi.org/10.1016/j.apal.2025.103665}

\bibitem{KT}
  Vera Koponen and Yasmin Tousinejad,
  Random expansions of trees with bounded height,
  {\em Theoretical Computer Science},
  Vol. 1040, (2025) 115201,
  \url{https://doi.org/10.1016/j.tcs.2025.115201}.


\bibitem{KW1}
  Vera Koponen and Felix Weitkämper,
  Asymptotic elimination of partially continuous aggregation functions in directed graphical models,
  {\em Information and Computation},
  Vol. 293 (2023)
  105061,
  \url{https://doi.org/10.1016/j.ic.2023.105061}.

\bibitem{KW2}
 Vera Koponen and Felix Weitkämper,
  On the relative asymptotic expressivity of inference frameworks,
  {\em Logical Metods in Computer Science},
  Vol. 20 (2024)
  13:1--13:52,
\url{https://doi.org/10.46298/lmcs-20(4:13)2024}.

\bibitem{LT}
  J. Lukasiewicz and A. Tarski,
  Untersuchungen uber den Aussagenkalkul,
  {\em Comptes Rendus des S\'{e}ances de la Soci\'{e}t\'{e} des Sciences et Lettres de Varsovie},
  Vol. 23 (1930)
  30--50.
  
\bibitem{Lyn} J. F. Lynch, Almost sure theories,
{\em Annals of Mathematical Logic}, Vol. 18 (1980) 91--135.

\bibitem{McC} Gregory L. McColm, MSO zero-one laws on random labelled acyclic graphs,
{\em Discrete Mathematics}
Vol. 254 (2002) 331--347.

\bibitem{MT}  D. Mubayi, C. Terry, Discrete metric spaces: structure, enumeration, and 0-1 laws, 
{\em Journal of Symbolic Logic}, Vol. 84 (2019) 1293--1324.

\bibitem{RS}
 Fabrizio Riguzzi and Theresa Swift,
  A survey of probabilistic logic programming,
  {\em Declarative Logic Programming: Theory, Systems, and Applications}
  (2018)
  185--228.

\bibitem{SS} S. Shelah, J. Spencer, Zero-one laws for sparse random graphs,
{\em Journal of the American Mathematical Society}, Vol. 1 (1988) 97--115.


\bibitem{Wei21}
  Felix Weitkämper,
  An asymptotic analysis of probabilistic logic programming,
  {\em Theory and Practice of Logic Programming},
  Vol. 21 (2021)
  802--817.


\bibitem{Wei24}
  Felix Weitkämper,
  Functional Lifted Bayesian Networks: Statistical Relational Learning and Reasoning with Relative Frequencies,
  {\em Inductive Logic Programming. ILP 2022. Lecture Notes in Computer Science}, Vol. 13779 (2024)
 142--156.

\end{thebibliography}
\end{document}